\def\gcd{\mathop{\rm gcd}}
\def\diag{\mathop{\rm diag}}
\def\rank{\mathop{\rm rank}}
\def\MM{\mathop{\mathrm M}\nolimits}
\def\amc{\mathop{\text{\small\sc AMC}}\nolimits}
\def\mbc{\mathop{\text{\small\sc MBC}}\nolimits}

\def\wgt{\mathop{\rm wgt}\nolimits}

\def\ker{\mathop{\rm Ker}}
\def\im{\mathop{\rm im}}

\def\css{\mathop{\rm CSS}\nolimits}

\def\bs#1{\boldsymbol{#1}}

\documentclass[aps,pra,reprint,superscriptaddress,
tightenlines,longbibliography]{revtex4-2}

\usepackage{graphicx}
\graphicspath{{./}{./figs/}}
\usepackage{braket}
\usepackage{amsmath,amssymb,amsfonts}
\usepackage{amsthm}
\newtheorem{lemma}{Lemma}
\newtheorem{statement}[lemma]{Statement}

\usepackage{multirow}
\usepackage[colorlinks,citecolor=blue]{hyperref}
\usepackage[dvipsnames]{xcolor}

\usepackage[T2A]{fontenc}
\usepackage[russian,english]{babel}
\usepackage[utf8]{inputenc}
\def\nz{{\sf И}-{\sf Z}}

\usepackage[normalem]{ulem}

\usepackage{enumitem}
\setlist{nosep}

\usepackage{tikz}

\preprint{arXiv:2506.16910}
\begin{document}

\title{Abelian multi-cycle codes for single-shot error correction.}

\author{Hsiang-Ku Lin}
\affiliation{Department of Physics \& Astronomy, University of
  California, Riverside, California 92521 USA}

\author{Pak Kau Lim}

\affiliation{Department of Physics \& Astronomy, University of
  California, Riverside, California 92521 USA}

\author{Alexey A. Kovalev}

\affiliation{Department of Physics and Astronomy and Nebraska Center
  for Materials and Nanoscience, University of Nebraska, Lincoln,
  Nebraska 68588, USA}

 \author{Leonid P. Pryadko}
 \email{pryadko@google.com}
 \affiliation{Google Quantum AI,  Santa Barbara, California 93117, USA}
 \affiliation{Department of Physics \& Astronomy, University of
   California, Riverside, California 92521 USA}

 \date{\today}
\begin{abstract}
  We construct a family of quantum low-density parity-check codes
  locally equivalent to higher-dimensional quantum hypergraph-product
  (QHP) codes.  Similarly to QHP codes, the proposed codes have highly
  redundant sets of low-weight stabilizer generators, which improves
  decoding accuracy in a fault-tolerant regime and gives them
  single-shot properties.  The advantage of the new construction is
  that it gives shorter codes.  We derive simple expressions for the
  dimension of the proposed codes in two important special cases, give
  bounds on the distances, and explicitly construct some relatively
  short codes.  Circuit simulations for codes locally equivalent to
  4-dimensional toric codes show a (pseudo)threshold close to $1.1\%$,
  better than for toric or surface codes with a similar noise model.
\end{abstract}
\maketitle

\section{Introduction.}

Quantum error correction (QEC) is one of the key enabling technologies
for scalable quantum computation.  Despite a lot of effort in
the community, it took over a quarter century after discovery of
QEC\cite{Shor-FT-1996,gottesman-thesis} for quantum memory experiments
based on surface codes\cite{kitaev-anyons,Bravyi-Kitaev-1998,%
  Dennis-Kitaev-Landahl-Preskill-2002} to demonstrate logical error rates that
are well in the scalable
regime\cite{google-2023-suppressing,Paetznik-etal-trapped-ions-2024,%
  Bluvstein-etal-Lukin-2024,google_Quantum_AI-Y1-2025}.  A major
difficulty here is the need to operate fault-tolerantly, in the
presence of errors that occur during the syndrome measurement.

A general approach to fault-tolerance (FT) is to use redundant
syndrome measurements.  For quantum memory and logical Clifford gates,
the redundancy can be gained from repeated measurements, as it is
sufficient to track just the Pauli basis associated with the current
error.  The resulting classical syndrome-based decoding problem can be
solved sequentially, using measured syndrome data from a few cycles at
a
time\cite{Dennis-Kitaev-Landahl-Preskill-2002,Gong-Cammerer-Renes-2024}.
On the other hand, non-Clifford gates needed for universal quantum
computation require real-time decoding.  Even though some delay can be
tolerated at the cost of increased circuit complexity, shorter lag
times are strongly
preferable\cite{Skoric-Browne-Barnes-Gillespie-Campbell-2023}.  From
this viewpoint, optimal is single-shot
decoding\cite{Bombin-2015,Brown-Nickerson-Browne-2016,Campbell-2018},
where syndrome data from each measurement cycle is processed right
after it becomes available.

Single-shot fault-tolerant QEC requires measuring some redundant
stabilizer generators.  This redundancy helps to control measurement
errors and thus enables independent decoding of each batch of syndrome
data.  Such ideas go back to the seminal
work\cite{Dennis-Kitaev-Landahl-Preskill-2002} of Dennis et al.;
various coding theory aspects have been studied in, e.g.,
Refs.~\onlinecite{Fujiwara-2014,Ashikhmin-Lai-Brun-2014,%
  Ashikhmin-Lai-Brun-2016}.  Single-shot properties have been studied
numerically and analytically for a number of code families, including
higher-dimensional hyperbolic codes\cite{Breuckmann-Londe-2020}, as
well as higher-dimensional
toric\cite{Breuckmann-Duivenvoorden-Michels-Terhal-2017}, and more
general quantum hypergraph-product (QHP)
codes\cite{Zeng-Pryadko-2018,%
  Zeng-Pryadko-hprod-2020,%
  Quintavalle-Vasmer-Roffe-Campbell-2021,Higgott-Breuckmann-2023}.

Unfortunately, the product constructions tend to give rather long
codes.  Recently, the effect of natural redundancy in two-block and
generalized-bicycle (GB) codes\cite{Kovalev-Pryadko-Hyperbicycle-2013}
with stabilizer generators of weight six has been studied by three of
the present authors\cite{Lin-Liu-Lim-Pryadko-circuits-2024}.  These,
and more general 2BGA
codes\cite{Wang-Pryadko-2022,Wang-Lin-Pryadko-2023,Lin-Pryadko-2023},
including the bivariate-bicycle (BB)
codes\cite{Bravyi-etal-Yoder-2023,Liang-Liu-Song-Chen-2025,%
  Berthusen-etal-Gottesman-2025,%
  Wang-etal-Deng-2025}, have attracted attention as the
family which contains short codes with relatively high rates
and distances (substantially improving over the surface/toric codes)
and at the same time with relatively high thresholds.  While GB codes
constructed in Ref.~\onlinecite{Lin-Liu-Lim-Pryadko-circuits-2024}
saturate an upper bound on the syndrome distance, their single-shot
properties turned out to be limited by the lack of
confinement.

More recently, Aasen et al.\ proposed a family of four-dimensional
geometric toric
codes\cite{Aasen-Haah-Hastings-Wang-2025,Aasen-etal-Svore-2025}
locally equivalent to 4D toric codes, but on a 4-torus with
periodicity vectors not necessarily along the principal lattice
directions.  They analyzed the performance of the constructed codes in
circuit simulations of quantum memory experiments, and studied their
single-shot properties.  They also designed a number of protected
operations on logical qubits, including state injection and the
unitaries generating the entire Clifford group on the logical qubits.

In this work, we construct a family of multi-block chain (MBC)
complexes, and their implementation using abelian group algebra
matrices, which we call abelian multi-cycle (AMC) codes.  This
construction is strictly more general than 4D geometric toric codes in
Refs.~\onlinecite{Aasen-Haah-Hastings-Wang-2025,Aasen-etal-Svore-2025},
in the sense that all their codes are included in our construction.
Our MBC codes are related to higher-dimensional QHP
codes\cite{Zeng-Pryadko-2018,Zeng-Pryadko-hprod-2020} in exactly the
same way as two-block codes are related to the original
(two-dimensional) QHP codes\cite{Tillich-Zemor-2009}, and they share
many properties with both the two-block codes and higher-dimensional
QHP codes.  In addition to the general construction, we give analytic
expressions for the code dimensions in two important general cases,
and also give lower (existence) and upper bounds for the distances.
As an important example, we study the family of AMC codes locally
equivalent to 4-dimensional toric codes.  All these codes have
stabilizer generators of weight $6$ and encode $k=6$ logical qubits.
We explicitly construct a number of such codes, calculate their
parameters, and perform circuit simulations to compute a
(pseudo)threshold and analyze the accuracy of few-shot sliding-window
(SW) decoding.

\section{Definitions.}
For a given finite (Galois) field $F\equiv \mathbb{F}_q$, of size
$|F|=q=p^m$ and prime characteristic $p$, an $F$-linear code of length
$n$ and dimension $k$ is a linear space ${\cal C}\subseteq F^n$ of
dimension $k$.  Any set of $k=\rank G$ linearly-independent rows of an
$F$-valued \emph{generator} matrix $G$ with $n$ columns forms a basis
of the code ${\cal C}\equiv {\cal C}_G$ generated by $G$; this code
has length $n$ and dimension $k$.  The dual code
${\cal C}^\perp\equiv{\cal C}_G^\perp$ of dimension $n-k$ is formed by
all vectors in $F^n$ orthogonal to non-zero vectors in ${\cal C}$ or,
equivalently, to the rows of $G$, a \emph{check}
matrix of the code ${\cal C}_G^\perp$.

For a given finite field $F$ and a finite group $G$ of order
$|G|=\ell$, the (Hopf) \emph{group algebra} (also, a ring) $F[G]$ is
defined as an $F$-linear space of all formal sums
\begin{equation}
  \label{eq:algebra-element}
  x\equiv \sum_{g\in G}x_g g,\quad x_g\in F,
\end{equation}
where group elements $g\in G$ serve as basis vectors,
equipped with the product naturally associated with the group
operation,
\begin{equation}
  \label{eq:FG-product}
  ab=\sum_{g\in G}\biggl(\sum_{h\in G} a_h b_{h^{-1}g}\biggr) g, \quad a,b\in F[G].
\end{equation}
Evidently, Eq.~(\ref{eq:algebra-element}) defines a one-to-one map
between any vector $\bs x\in F^\ell$ with coefficients $x_g$ labeled
by group elements and a group algebra element $x\in F[G]$, and a
related map between sets of vectors and sets of group algebra
elements.  In this work we consider only group algebras constructed
from abelian groups.  An abelian group code in $F^\ell$ is a map
(\ref{eq:algebra-element}) of an ideal $J$ in an abelian ring $F[G]$,
defined as an $F$-linear space of elements of $F[G]$ such that for any
$x\in J$ and any $r\in F[G]$, $rx\in J$.  An ideal $J\le F[G]$ is
called\cite{Borello-delaCruz-Willems-2022} {\em checkable\/} if there
is an element $a\in F[G]$ such that $J=\{x| x\in F[G], a x=0\}$.

A \emph{chain complex} is a sequence of abelian groups and
homomorphisms between pairs of consecutive groups such that the image
of each homomorphism is included in the kernel of the next.  Here we
will be concerned with the special case of chain complexes of
finite-dimensional vector spaces
$\ldots,\mathcal{A}_{j-1}, \mathcal{A}_j,\ldots$ over a finite field
$F$, and we denote the space dimensions
$n_j\equiv n_j({\cal A})=\dim {\cal A}_j$.  In this case the boundary
operators are linear transformations
$\partial_j:{\cal A}_{j-1}\leftarrow {\cal A}_j $ that map between
each pair of neighboring spaces, with the requirement
$\partial_j\partial_{j+1}=0$, $j\in\mathbb{Z}$.  We define a
$D$-complex ${\cal A}\equiv {\cal K}(A_1,\ldots,A_D)$, a bounded chain
complex which only contains $D+1$ non-trivial spaces with fixed bases,
in terms of $n_{j-1}\times n_j$ matrices $A_j$ over $F$ serving as the
boundary operators, $j\in\{1,\ldots,D\}$:
\begin{equation}
  \label{eq:chain-complex}
  {\cal A}:\;
  \ldots \leftarrow\{0\}\stackrel{\partial_0}\leftarrow {\cal
    A}_0\stackrel{A_1}\leftarrow
  {\cal A}_1\ldots \stackrel{A_{D}}\leftarrow
  {\cal A}_{D}\stackrel{\partial_{D+1}}\leftarrow
  \{0\} \ldots
\end{equation}
Here the products of neighboring matrices must be zero,
$A_{j-1}A_{j}=0$, $j\in\{2,\ldots,D\}$.  In addition to boundary
operators given by the matrices $A_j$, implicit are the trivial
operators $\partial_0:\{0\}\leftarrow {\cal A}_0$ and
$\partial_{D+1}: {\cal A}_D\leftarrow \{0\}$ (with the image being the
zero vector in ${\cal A}_D$) treated formally as rank-zero
$0\times n_0$ and $n_D\times 0$ matrices.

Elements of the subspace $\im (\partial_{j})\subseteq {\cal A}_{j-1}$
are called boundaries; in our case these are linear combinations of
columns of $A_{j}$ and, therefore, form an $F$-linear code with the
generator matrix $A_{j}^T$, $\im (A_{j})=\mathcal{C}_{A_{j}^T}$.
Elements of $\ker(\partial_j)\subset {\cal A}_j$ are called cycles; in
our case these are vectors in an $F$-linear code with the parity check
matrix $A_j$, $\ker (A_j)=\mathcal{C}_{A_j}^\perp$.  

Because of the orthogonality $\partial_j\partial_{j+1}=0$, all
boundaries are necessarily cycles,
$\im(\partial_{j+1})\subseteq \ker(\partial_j) \subseteq {\cal A}_j$.
The structure of the cycles in ${\cal A}_j$ that are not boundaries is
described by the $j$\,th homology group,
\begin{equation}
  {H}_j({\cal A})\equiv H(A_j,A_{j+1})=
  \ker(A_{j})/\im(A_{j+1}).
\label{eq:homo-group}
\end{equation}
Group quotient here means that two cycles that differ by a boundary
are considered equivalent; non-zero elements of
$\mathcal{H}_j(\mathcal{A})$ are equivalence classes of homologically
non-trivial cycles.  Explicitly, the equivalence of $\bs x$ and
$\bs y$ in $\mathcal{A}_j$, denoted as $\bs x\simeq \bs y$, implies
that for some $\bs\alpha\in {\cal A}_{j+1}$,
$\bs y=\bs x+ \bs\alpha A_{j+1}^T$.  The rank of $j$-th homology group
is the dimension of the corresponding vector space; one has
\begin{equation}
  \label{eq:homo-rank}
  k_j\equiv \rank H_j(\mathcal{A})=n_j-\rank A_j-\rank A_{j+1}  .
\end{equation}
The homological \emph{distance} $d_j$ is the minimum Hamming weight of
a non-trivial element (any representative) in the homology group
$H_j(\mathcal{A})\equiv H(A_j,A_{j+1})$,
\begin{equation}
  d_j=\min_{ \bs 0\not\simeq \bs x\in H_j(\mathcal{A})} \wgt \bs x
  =\min_{\bs x\in \ker({A}_j)\setminus \im(A_{j+1})}\wgt \bs x.
  \label{eq:homo-distance}
\end{equation}
By this definition, $d_j\ge1$.  To address singular cases, the minimum
of an empty set is defined as infinity; $k_j=0$ is always equivalent
to $d_j=\infty$.  In particular, the distance of the homology group
$H_0(\mathcal{A})$ is $d_0=1$, unless $A_1$ has full row rank, giving
$k_0=0$, in which case $d_0=\infty$.  In the case of the
homology group $H_{D}(\mathcal{A})$, the distance $d_{D}$ is that of
the $F$-linear code with the check matrix $A_D$.

In addition to the homology group $H(A_j,A_{j+1})$, there is also a
\emph{co-homology} group
$\tilde{H}_j(\tilde{\cal A})=H(A_{j+1}^T,A_j^T)$ of the same rank
(\ref{eq:homo-rank}); this is associated with the \emph{co-chain
  complex} $\tilde{\cal A}$ formed from the transposed matrices
$A_j^T$ taken in the opposite order.  A quantum CSS
code\cite{Calderbank-Shor-1996,Steane-1996} with generator matrices
$G_X=A_j$ and $G_Z=A_{j+1}^T$ is isomorphic with the direct sum of the
groups $H_j$ and $\tilde{H}_j$,
\begin{equation}
  \label{eq:css-code-homology}
\css(A_j,A_{j+1}^T)\cong H(A_j,A_{j+1})\oplus H(A_{j+1}^T,A_j^T).
\end{equation}
The two terms correspond to $Z$ and $X$ logical operators,
respectively, with stabilizer generator matrices $H_X=A_j$ and
$H_Z=A_{j+1}^T$.  This gives for the homological distances in the
chain complex and in the co-chain complex, respectively, $d_j=d_Z$ and
$\tilde{d}_j=d_X$.

Quantum codes operate in a fault-tolerant (FT) regime, where errors
may happen during the syndrome measurement.  Commonly used in such a
situation is the tactic of redundant measurements.  Assuming each
operator being measured has the corresponding row in $H_X$ or $H_Z$,
redundancy implies the existence of non-trivial matrices $M_X$ and
$M_Z$ such that $M_X H_X=0$, $M_Z H_Z=0$.  Rows of $M_X$ and $M_Z$,
respectively, correspond to $X$ and $Z$
metachecks\cite{Campbell-2018}.  Evidently, a quantum CSS code with
redundant measurements forms a $4$-chain complex with boundary
operators $M_X$, $H_X$, $H_Z^T$, $M_Z^T$.  In particular, Campbell et
al.\ argued\cite{Campbell-2018,Quintavalle-Vasmer-Roffe-Campbell-2021}
that longer complexes constructed as tensor products of $1$-complexes,
have nice error-correction properties in the FT regime and they can be
used for single-shot error correction.

Tensor product $\mathcal{A}\times \mathcal{B}$ of two chain complexes
$\mathcal{A}$ and $\mathcal{B}$ is defined as the chain complex
formed by linear spaces decomposed as direct sums of Kronecker
products,%
\begin{equation}
  (\mathcal{A}\times \mathcal{B})_j
  =\bigoplus\nolimits_{i\in\mathbb{Z}}\mathcal{A}_i \otimes
  \mathcal{B}_{j-i},\label{eq:tp-spaces}
\end{equation}
with the action of the boundary operators
\begin{equation}
  \label{eq:tp-boundary}
  \partial'''(\bs a\otimes \bs b)\equiv \partial' \bs a\otimes \bs b+(-1)^i \bs a\otimes
  \partial'' \bs b,
\end{equation}
where $\bs a\in\mathcal{A}_i$, $\bs b\in\mathcal{B}_{j-i}$, and the boundary
operators $\partial'$, $\partial''$, and $\partial'''$ act in
complexes $\mathcal{A}$, $\mathcal{B}$, and
$\mathcal{A}\times \mathcal{B}$, respectively.  Notice that the two
terms in Eq.~(\ref{eq:tp-boundary}) are supported in different
subspaces of the expansion (\ref{eq:tp-spaces}).  When both
$\mathcal{A}$ and $\mathcal{B}$ are \emph{bounded}, the dimension
$n_j(\mathcal{C})$ of a space $\mathcal{C}_j$ in the product
$\mathcal{C}=\mathcal{A}\times \mathcal{B}$ is%
\begin{equation}
  \label{eq:tp-nj}
  n_j (\mathcal{C})=\sum\nolimits_{i} n_i (\mathcal{A}) \,n_{j-i}
  (\mathcal{B}).
\end{equation}
The homology groups of the product $\mathcal{C}={\cal A}\times {\cal B}$ are
isomorphic to a simple expansion in terms of those of ${\cal A}$ and
${\cal B}$ which is given by the K\"unneth formula,
\begin{equation}
  \label{eq:tp-Kunneth}
  H_j (\mathcal{C})\cong\bigoplus\nolimits_{i} H_i (\mathcal{A})\,
  \otimes\,H_{j-i}
  (\mathcal{B}).
\end{equation}
One immediate consequence is that the rank $k_j(\mathcal{C})$ of the
$j$\,th homology group $H_j(\mathcal{C})$ is
\begin{equation}
  \label{eq:tp-kj}
  k_j (\mathcal{C})=\sum\nolimits_{i} k_i (\mathcal{A}) \,k_{j-i}
  (\mathcal{B}).
\end{equation}

For the homological distances in a product of two
arbitrary-length complexes, we have some
bounds\cite{Zeng-Pryadko-hprod-2020},
\begin{equation}
  \label{eq:upper-general}
  \min_{i\in\mathbb{Z}}  \max\biglb( d_i(\mathcal{A}),d_{j-i}(\mathcal{B})\bigrb)
  \le
  d_{j}(\mathcal{C})\le \min_{i\in\mathbb{Z}} d_{i}(\mathcal{A})
  d_{j-i}(\mathcal{B}).
\end{equation}
The exact distances are known only in the important case where one of
the factors is a $1$-complex\cite{Tillich-Zemor-2009,Zeng-Pryadko-2018,%
  Zeng-Pryadko-hprod-2020}.  With a $D$-complex $\mathcal{A}$ in
Eq.~(\ref{eq:chain-complex}) and a $1$-complex
$\mathcal{B}\equiv \mathcal{K}(B_1)$, the homological distances
are\cite{Zeng-Pryadko-hprod-2020}%
\begin{equation}
  \label{eq:result-thm}
  d_j(\mathcal{C})
  =\min\biglb( d_j(\mathcal{A}) d_0(\mathcal{B}),
  d_{j-1}(\mathcal{A}) d_1(\mathcal{B})\bigrb).
\end{equation}

In the following we refer to a $D$-fold product of $1$-complexes as a
$D$-dimensional QHP
complex\cite{Tillich-Zemor-2009,Zeng-Pryadko-2018,Zeng-Pryadko-hprod-2020}.
Such a $D$-complex ${\cal Q}$ has $D+1$ non-trivial linear spaces
${\cal Q}_j$, with $0\le j\le D$, where the space at level $j$ is
formed as a direct sum of $m_j^{(D)}$ product subspaces,
$${\cal A}^{(1)}_{j_1}\times{\cal A}^{(2)}_{j_2}\times \ldots \times
{\cal A}^{(D)}_{j_D}.$$ Here the level index $j=j_1+j_2+\ldots+ j_D$
is a sum of individual level indices $j_i\in \{0,1\}$.  The
number of such terms is given by a binomial,
\begin{equation}
  m_j^{(D)}=\binom{D}{j}.\label{eq:qhp-block-dim}
\end{equation}
In particular, consider a $1$-complex ${\cal A}\equiv\mathcal{K}(A)$,
where $A$ is a square $n\times n$ matrix such that the $F$-linear code
${\cal C}_{A_1}^\perp$ has parameters $[n,k,d]$.  Its $D$-fold tensor
product with itself, $\mathcal{Q}={\cal A}^{\times D}$, has space dimensions
\begin{equation}
  \label{eq:QHP-power-nj}
  n_j^{(D)}=\binom{D}{j} n^D,\quad 0\le j\le D,
\end{equation}
and the dimension and the homological distance of its $j$\,th homology
group $H_j(\mathcal{Q})$ are, respectively,
\begin{equation}
  \label{eq:QHP-power-kj_dj}
  k_j^{(D)}=\binom{D}{j} k^D,\quad 
  d_j^{(D)}= d^j,\quad 0\le j\le D.
\end{equation}

\section{Construction and code properties.}
\subsection{Multi-block complexes and associated quantum codes}
We construct a \emph{multi-block chain} (MBC) complex as a
generalization of $D$-dimensional QHP codes in exactly the same
fashion as the two-block codes have been constructed as a
generalization of regular (2D) HP
codes\cite{Kovalev-Pryadko-Hyperbicycle-2013}.  Namely, given matrices
$A_1$ and $B_1$ over $F$, with dimensions
$n_0(\mathcal{A})\times n_1(\mathcal{A})$ and
$n_0(\mathcal{B})\times n_1(\mathcal{B})$, respectively, the tensor
product $\mathcal{P}=\mathcal{A}\times \mathcal{B}$ of the complexes
$\mathcal{A}\equiv \mathcal{K}(A_1)$ and
$\mathcal{B}\equiv \mathcal{K}(B_1)$, explicitly,
$$
\mathcal{P}:
\ldots\{0\}\stackrel{\partial_0}\leftarrow {\cal A}_0\times {\cal B}_0
\stackrel{P_1}\leftarrow
\begin{array}[c]{c}
  {\cal A}_0\times{\cal B}_1\\   {\cal A}_1\times{\cal B}_0
\end{array}
\stackrel{P_2}\leftarrow {\cal A}_1\times{\cal B}_1
\leftarrow
  \{0\} \ldots,
$$
has the boundary operators
[cf.~Eq.~(\ref{eq:tp-boundary})]
\begin{equation}
  \label{eq:2D-matrices-kp}
  {P}_1=  \left[I_A\otimes B_1,A_1\otimes I_B\right],\quad
  {P}_2 =\left[
    \begin{array}[c]{c}
      A_1\otimes I_B \\-I_A\otimes B_1
    \end{array}
  \right],
\end{equation}
where $I_A\equiv I(\mathcal{A}_0)$ and $I_B\equiv I(\mathcal{B}_0)$
are the identity operators on the spaces $\mathcal{A}_0$ and
$\mathcal{B}_0$, respectively.  The product vanishes, $P_1P_2=0$,
because the two blocks, $A\equiv A_1\otimes I_B$ and
$B\equiv I_A\otimes B_1$, commute.  The
generalization\cite{Kovalev-Pryadko-Hyperbicycle-2013} is evident: for
any pair of $\ell\times \ell$ square commuting matrices, $AB=BA$, the
two-block code has the CSS matrices
\begin{equation}
  \label{eq:2D-matrices}
  H_X\equiv P_1=\left[B,A\right],\quad H_Z^T\equiv P_2=\left[
    \begin{array}[c]{c}
      A\\-B
    \end{array}\right].
\end{equation}
We regard this construction as the smallest (two-dimensional)
multi-block chain complex, and denote it ${\cal P}=\mbc(A,B)$, with
the requirement $AB=BA$.

Similarly, to define a three-dimensional generalization, a
\emph{3-chain complex} $\mbc(A,B,C)$ from three pairwise commuting
matrices, we insert an additional $1$-complex into the product,
$\mathcal{Q}\equiv \mathcal{A}\times \mathcal{B}\times
\mathcal{C}=\mathcal{P}\times {\cal C}$, with the non-trivial
subspaces
\begin{eqnarray*}
  \mathcal{P}_0\times \mathcal{C}_0
  &\stackrel{Q_1}\leftarrow&
               \begin{array}[c]{c}
                 \mathcal{P}_0\times \mathcal{C}_1\\
                 \mathcal{P}_1\times \mathcal{C}_0\\
               \end{array}\stackrel{Q_2}\leftarrow
  \begin{array}[c]{c}
    \mathcal{P}_1\times \mathcal{C}_1\\
    \mathcal{P}_2\times \mathcal{C}_0
  \end{array}\stackrel{Q_3}\leftarrow
               \mathcal{P}_2\times \mathcal{C}_1.
\end{eqnarray*}
The subspaces in direct sums are stacked vertically
[cf.~Eq.~(\ref{eq:tp-spaces})], and ordered alphabetically by the
space indices.  The structure of the corresponding boundary operators
follows directly from the subspace decomposition and the sign
convention~(\ref{eq:tp-boundary}),
$Q_1= \left[I\otimes C_1, P_1\otimes I\right]$,%
$$
Q_2= \left[
  \begin{array}[c]{cc}
    P_1\otimes I & 0 \\
    -I\otimes C_1& {P}_2 \otimes I
  \end{array}
\right],\;\,\text{and}\,\;
Q_3= \left[
  \begin{array}[c]{c}
    P_2\otimes I \\ I\otimes C_1
  \end{array}\right].
$$
With Eq.~(\ref{eq:2D-matrices-kp}), it is easy to verify that the
resulting matrices contain only three blocks,
$A\equiv A_1\otimes I_B\otimes I_C$,
$B\equiv I_A\otimes B_1\otimes I_C$, and
$C\equiv I_A\otimes I_B\otimes C_1$, and can be explicitly written as
\begin{eqnarray}
  Q_1&=&[C\,|\,B\; A],\label{eq:mbc3-1}\\
  Q_2&=&\left[
  \begin{array}[c]{cc|c}
    B&A&0\\ \hline -C&0 &A\\ 0& -C& -B
  \end{array}
  \right], \label{eq:mbc3-2} \\
  Q_3&=&\left[
  \begin{array}[c]{c}
    A\\-B\\ \hline C
  \end{array}
  \right].\label{eq:mbc3-3}
\end{eqnarray}
We used the lines here to emphasize the inherited block structure,
cf.~matrices in Eq.~(\ref{eq:2D-matrices}).  Just like for the case of
two-block codes, the orthogonality can be verified directly,
$Q_1Q_2=0$, $Q_2Q_3=0$, as long as the individual blocks $A$, $B$, and
$C$ commute.

For a general definition, consider a $D$-complex
${\cal Q}\equiv\mbc(A_1,A_2,\ldots,A_D)$
constructed from $D$ commuting square blocks $A_i$, $\ell\times\ell$
matrices over $F$, with boundary operators $\partial_j^{(Q)}=Q_j$,
where $Q_jQ_{j+1}=0$.  Explicitly, the space dimensions (non-zero for
$0\le j\le D$) are
$\dim({\cal Q}_j)\equiv n_j^{(Q)}=\binom{D}{j}\ell$; it is also
convenient to define the block dimensions $m_j\equiv m_j^{(Q)}=\binom{D}{j}$
such that $n_j=m_j\ell$. Then, given an additional block
$A_{D+1}\equiv N$ commuting with those already present, the
$(D+1)$-dimensional chain complex
${\cal R}\equiv\mbc(A_1,A_2,\ldots,A_D,N)$ has the block dimensions
$m_{j}^{(R)}=m_{j-1}^{(Q)}+m_j^{(Q)}=\binom{D+1}{j}$, and the boundary
operators constructed as block matrices,
\begin{eqnarray}
  \label{eq:mbcD-1}
R_1 &=& \bigl[
\begin{array}{cc}
     N, & Q_1
\end{array}
\bigr],\\
R_i &=&\left[
\begin{array}{cc}
Q_{i-1}, &0 \\
(-1)^{i-1} I({m_{i-1}})\otimes N, & Q_i
\end{array}
\right],\;\, 1< i\le D,\qquad \label{eq:mbcD-2}\\
R_{D+1} &=& \left[
\begin{array}{c}
Q_D \\   (-1)^{D} N
\end{array}
\right], \label{eq:mbcD-3}
\end{eqnarray}
where $I(m)$ is an $m\times m$ identity matrix.

For a $D$-dimensional MBC complex with boundary matrices $Q_j$,
$1\le j\le D$, it is easy to verify that all $D$ blocks are used in an
equivalent fashion, i.e., all chain complexes constructed from the
same $D$ commuting blocks taken in different orders are equivalent
(can be obtained from each other by block row and column
permutations).  In particular, matrix $Q_j$ has $j$ non-zero blocks in
each column and $D+1-j$ non-zero blocks in each row.  Furthermore, in
each pair of matrices $Q_{j}$ and $Q_{j+1}^T$ with
$m_{j}=\binom{D}{j}$ column blocks, $1\le j< D$, each pair of matching
column-blocks contains a full set of blocks $\{A_i\}$, exactly $j$
such blocks in $Q_{j}$ and the remaining $D-j$ (transposed) blocks in
$Q_{j+1}^T$.

For future reference, the boundary operators in the 4-complex
${\cal R}\equiv\mbc(A,B,C,D)$ are
\begin{eqnarray}
  \label{eq:mbc4-1}
  R_1&=&\,\biglb[D\,|\,C\,B\,A\bigrb]\,,\\ \label{eq:mbc4-2}
  R_2&=&\left[
         \begin{array}[c]{ccc|ccc}
           C& B& A \\ \hline
           -D&&&B& A&0\\&-D&& -C&0&A\\ &&-D&0&-C&-B\\
         \end{array}
  \right],\\    \label{eq:mbc4-3}
  R_3&=&\left[
         \begin{array}[c]{ccc|c}
           B& A&0\\-C&0&A\\  0&-C&-B\\ \hline
           D&&&    A\\&D&&-B\\ &&D& C\\
         \end{array}
  \right],\;\,\,\\ \label{eq:mbc4-4}
  R_4&=& \left[\begin{array}[c]{ccc|c}A\\-B\\C\\ \hline -D
               \end{array}
  \right].  
\end{eqnarray}
While any non-trivial homology group in this complex may be used to
construct a quantum code, we will focus on
${H}_2(\mathcal{R})=H(R_2,R_3)$ and the corresponding co-homology
group $\tilde H_2({\cal R})$.  This gives a self-dual CSS code with
generator matrices $H_X=R_2$, $H_Z=R_3^T$.

\subsection{Abelian multi-cycle (AMC) codes.}
AMC complexes (and associated \emph{quasi-abelian} quantum CSS codes)
are a special case of multi-block codes where the commuting matrices
are constructed\footnote{Commuting matrices can be constructed for any
  pair of group algebra elements even for a non-abelian
  group\cite{Panteleev-Kalachev-2021}. While there is no direct
  generalization for arbitrary sets of $D>2$ non-abelian group algebra
  elements, less general constructions do exist.  These go outside the
  scope of the present work.} with the help of an abelian group
algebra $F[G]$, see Eq.~(\ref{eq:algebra-element}).  Namely, assuming
the group size is $\ell\equiv |G|$, for an element $a\in F[G]$, the
$\ell\times\ell$ matrix $A\equiv \MM(a)$ is defined by the 
action of $a$ on the group elements,
\begin{equation}
  \label{eq:L-R-action}
  [\MM(a)]_{\alpha,\beta}\equiv \sum_{g\in G}a_g\delta_{\alpha,g\beta},
\end{equation}
where group elements $\alpha,\beta\in G$ are used to index rows and
columns, and $\delta_{\alpha,\beta}=1$ if $\alpha=\beta$ and $0$
otherwise is the Kronecker delta.  It is easy to verify that for group
elements $g\in G$, $\MM(g)$ are permutation matrices forming the
regular $F$-linear representation of $G$.  Also, since transposition
gives the inverse of a permutation matrix, for any $a\in F[G]$, a
transposed matrix $[M(a)]^T$ can be obtained from the group algebra
element\cite{Borello-delaCruz-Willems-2022}
\begin{equation}
  \widehat a\equiv\sum_{g\in G}a_{g^{-1}} g=\sum_{g\in G}a_g g^{-1}, \quad
  [M(a)]^T=M(\widehat a).
  \label{eq:hat-operator}
\end{equation}

Then, given a sequence $\{a_j\}_{j=1}^D\subset F[G]$ of abelian group
algebra elements, with $D\ge2$, the corresponding matrices
$A_j=\MM(a_j)$ commute, and the $D$-complex
$\mathcal{Q}\equiv \mathcal{Q}^{(D)}=\mbc(A_1,A_2,\ldots,A_D)$ can be
constructed as explained in the previous section.  We denote such a
$D$-complex ${\cal Q}= \amc(a_1,a_2,\ldots,a_D)$.  The
corresponding level-$j$ quantum CSS code $\amc_j(a_1,a_2,\ldots,a_D)$,
$ j\le D$, is constructed from the level-$j$ homology
$H_j(\mathcal{Q})$ and co-homology $\tilde H_j(\mathcal{Q})$ groups.
Explicitly,
$$
\amc_j(a_1,a_2,\ldots,a_D)=\css(Q_{j},Q_{j+1}^T),
$$
where $Q_{j}$ is the $j$-th boundary operator (matrix) in the
$D$-complex $\mathcal{Q}$.

To visualize, consider a finite abelian group
$$G=\langle x_1,x_2,\ldots,x_m|r_1(\bs x),r_2(\bs x),
\ldots,r_m(\bs x)\rangle_{\rm abelian},
$$ presented with $m$ commuting generators $x_1$, $x_2$,
\ldots, $x_m$, and $m$ additional independent (monomial) relators
$$
1_G=r_i(\bs x)=\prod_{j=1}^m x_j^{\Delta_j}, \quad \Delta_j\equiv
{\Delta_j^{(i)}},\quad i=1,2,\ldots,m.
$$
Further, elements of the free abelian group
$\langle x_1,x_2,\ldots, x_m\rangle$ are in a one-to-one
correspondence with points of $\mathbb{Z}^m$.  The exponents in each
relator $r_i(\bs x)$ form a basis vector
$\bs\Delta^{(i)}=(\Delta_1^{(i)},\Delta_2^{(i)},\ldots,\Delta_m^{(i)})$
of a lattice on $\mathbb{Z}^m$; individual basis vectors are combined
as columns of a lattice basis matrix $\Delta$.  Any two points on
$\mathbb{Z}^m$ connected by a lattice vector are identified.  Then,
elements of the group $G$ are in a one-to-one correspondence with
inequivalent points in $\mathbb{Z}^m$ which form an $m$-torus
${\cal T}_m$.  The number of inequivalent points (also, the number of
points in a primitive cell of the lattice) is
$n=\left|G\right|=\left|\det\Delta\right|$.  In the absence of
single-variable degree-two relators, the torus $\mathcal{T}_m$ is
covered by $\mathbb{Z}^m$, with a well-defined covering function
$f:\mathbb{Z}^m\to {\cal T}_m$ which is one-to-one in each ball of
some radius $r_{\rm inj}>0$ on $\mathbb{Z}^m$, the injectivity radius.

The support of an $m$-variate polynomial over $F$ corresponds to a set
of points on $\mathbb{Z}^m$; multiplication by a generator $x_j$ gives
a translation along the $j$-axis. Respectively, the support of a group
algebra element $a\equiv a(\bs x)$ corresponds to a set of points on
the torus ${\cal T}_m$.

A group presentation where each relator depends only on the
corresponding variable, $r_i=x_i^{\Delta_i}$, with the lattice basis
matrix $\Delta=\diag(\Delta_1,\Delta_2,\ldots, \Delta_m)$, corresponds
to a decomposition\footnote{Such a decomposition exists according to
  the Fundamental Theorem of finite abelian groups.} of the finite
abelian group as a direct product of cyclic groups,
$$G=C_{\Delta_1}\times C_{\Delta_2}\times \ldots \times C_{\Delta_m}.$$
In such a case we get a torus with periodicity vectors along the
Cartesian axes, and the group order is given by a product,
$|G|=\prod_{i=1}^m \Delta_i$.

We note that the number of generators in a presentation of the same
group may vary.  In particular, while any cyclic group
$C_n=\langle x|x^n\rangle$ has a single-generator presentation,
presentations with two or more variables can always be constructed for
a sufficiently large group.  For example,
\begin{eqnarray*}
  C_{15}
  &=&C_3\times C_5=\langle x,y|x^3,y^5 \rangle_{\rm abelian},\quad
      \Delta={
      \left(\begin{array}[c]{cc}3,&0\\ 0,&5\end{array}\right)},\\
  C_{15}
  &=&\langle x,y|x^4y,xy^4 \rangle_{\rm abelian},\quad
      \Delta={
      \left(\begin{array}[c]{cc} 4,&1\\ 1,&4\end{array}\right).}
\end{eqnarray*}
where the only condition is that the determinant of the 
matrix $\Delta$ matches the group order, $n=\left|\det \Delta\right|$.

\subsection{Equivalent AMC codes}
\label{sec:equiv}
Changing a group presentation, one can try to come up with an
equivalent complex, e.g., with better locality properties.  How much
freedom do we have?  To answer this question, let us first discuss the
symmetries of AMC construction.  The following statement is given
without a proof, as a generalization of Theorem 3 in
Ref.~\onlinecite{Lin-Pryadko-2023}:
\begin{statement}
  \label{th:permutation-equiv}
  For any sequence $\{a_j\}\equiv \{a_j\}_{j=1}^D\subset F[G]$, the $D$-complex
  $\amc(\{a_j\})$ is permutation-equivalent to the $D$-complex
  constructed from the modified sequence:
  \begin{enumerate}[label={\rm (\roman*)}]
  \item A permuted sequence of group algebra elements;
  \item \label{theorem3:1}$ \{\varphi(a_j)\}$, for any automorphism
    $\varphi:G\to G$;
  \item \label{theorem3:3}$\{ x_ja_j\}$, for any non-zero $x_j\in F$;
  \item \label{theorem3:4}$ \{a_j\alpha_j\}$, for any $\alpha_j\in G$.
  \end{enumerate}
  In addition, the sequence $\{\widehat a_j\}$, see
  Eq.~(\ref{eq:hat-operator}), gives an AMC complex isomorphic to the
  original co-chain complex,
  $\amc(\{\widehat{a}_j\})\simeq \widetilde{\amc}(\{a_j\})$.
\end{statement}

Therefore, we can always ``rescale'' the elements in the original set,
and choose $a_{j}=1+\ldots$, $j\le D$.  Further, we can try to
simplify the polynomials for a $D$-dimensional AMC complex, e.g., by
choosing group generators in the support of each element, e.g.,
$a_j=1+x_j+(\text{other monomials})$.  Clearly, the second
transformation changes the group presentation and the basis vectors
associated with the lattice.

For example, with weight-two elements $a_i=1+x_i$, $1\le i\le D$,
we get a complex locally equivalent to $D$-dimensional QHP codes.
Assuming the entire group can be generated this
way\footnote{Otherwise, the spaces in an AMC complex can be decomposed
  onto a direct sum of subspaces corresponding to cosets of the
  subgroup $G_X=\langle \{x_j\}\rangle$ in the original group $G$; see
  Sec.\ IV-C in Ref.~\onlinecite{Lin-Pryadko-2023}. In the abelian
  case these complexes are permutation-equivalent to a complex in the
  subgroup $G_X\le G$.},
$$
G=\left\langle \{x_j\}_{j=1}^D\right\rangle,
$$
we get a $D$-complex locally equivalent to $D$-dimensional toric
codes.  With a quantum CSS code associated with level-$j$ subspace,
for errors with irreducible weight smaller than the injectivity radius
$r_{\rm inj}$, this guarantees identical syndrome weights, that is,
identical confinement
profiles\cite{Quintavalle-Vasmer-Roffe-Campbell-2021} up to
$r_{\rm inj}$.  We also expect the code distance to grow as
$d={\cal O}(r_{\rm inj}^\alpha)$, with the exponent
$\alpha=\min(j,D-j)$ expected for level-$j$ $D$-dimensional toric
codes, see Statement \ref{th:k-d-rates} below.  Further, we can use a
measuring circuit designed for a regular level-$j$ $D$-dimensional
toric code: its validity is guaranteed as a local
property.\footnote{Note, however, that the circuit distance need not
  always coincide with that of the original code.  As an example,
  while any single-ancilla measurement circuit would work for a toric
  code\cite{Manes-Claes-2025}, rotated surface codes require a
  carefully designed \nz\ addressing scheme for
  fault-tolerance\cite{Tomita-Svore-2014}.}

In practice, especially if we want to search for AMC chain complexes
constructed from group algebra elements with weights higher than two
and groups with two or more generators, it may be more efficient to go
over small-weight polynomials and suitable periodicity vectors.  For
example, such a method was used in
Ref.~\onlinecite{Liang-Liu-Song-Chen-2025} to construct a number of
weight-six locally-planar GB (also known as BB)
codes.

\subsection{Matrix ranks for AMC construction.}
While the ranks of the matrices can be computed numerically in each
particular case (e.g., with the help of Gauss elimination algorithm
over the field $F$, or by analyzing the polynomial equations directly
by constructing the corresponding Gr\"obner
bases\cite{Postema-Kokkelmans-2025}), we have analytically computed
the corresponding ranks in two cases.

First, consider a cyclic group of order $\ell$, with
the group algebra $F[C_\ell]$ isomorphic to the ring
$F[x]/(x^\ell-1)$ of modular polynomials.  Here the ranks of
the boundary operators in the $D$-complex $\amc(a_1,a_2,\ldots,a_D)$
constructed from single-variate polynomials $a_j\equiv a_j(x)$,
$1\le j\le D$, is expressed in terms of the
\emph{characteristic polynomial},
\begin{equation}
h(x)=\gcd(a_1,a_2,\ldots,a_D,x^\ell-1).\label{eq:characteristic-poly}
\end{equation}
In Appendix \ref{sec:proof-k-cyclic}, we prove the following statement:
\begin{statement}
  \label{th:k-cyclic}
  Consider a $D$-complex $\amc(\{a_j\}_{j=1}^D)$ over $F[C_\ell]$.
  The ranks of the boundary operators satisfy
  $\rank Q_j=\binom{D-1}{j-1}(\ell-\kappa)$, where $\kappa=\deg h(x)$
  is the degree of the characteristic polynomial
  (\ref{eq:characteristic-poly}).
\end{statement}
As an immediate consequence, the corresponding homology groups have
orders $\left| H_j({\cal Q})\right|=\binom{D}{j}\kappa$.  Previously,
the same or equivalent results were derived for the special case of
$D=2$ AMC codes (quasi-cyclic GB codes), see
Refs.~\cite{Kovalev-Pryadko-Hyperbicycle-2013,Panteleev-Kalachev-2019}. Also
note a similarity with the corresponding equations for $D$-dimensional
QHP codes.  The main difference with 
Eqs.~(\ref{eq:QHP-power-nj}) and (\ref{eq:QHP-power-kj_dj}) is the linear
scaling with the group order $\ell$, which is why the AMC construction
gives much smaller codes.

Second is the case of a semisimple group algebra $F[G]$ where,
according to Maschke's theorem, the group order $\ell=\left|G\right|$
is mutually prime with the field characteristic $p$, $\gcd(p,\ell)=1$
(see, e.g., Corollary 2.2.5 in
Ref.~\onlinecite{Drozd-Kirichenko-book-1994}).  In this case the
result is very similar to that in Statement \ref{th:k-cyclic}, except
the value of $\kappa$ is defined differently.  For a technical
definition of $\kappa$ in terms of products of idempotent matrices,
please see the proof given in Appendix \ref{sec:proof-k-semisimple}.
In practice, $\kappa$ can be defined, e.g., as the dimension of the
classical $F$-linear code with the check matrix ${Q}_D$ (all $D$
square blocks stacked together in a single column block).
\begin{statement}
  \label{th:k-semisimple}
  Consider a $D$-dimensional complex
  ${\cal Q}=\amc(a_1,a_2,\ldots,a_D)$ over a semi-simple abelian group
  algebra $F[G]$.  The ranks of all boundary operators are
  proportional to that of $Q_D$,
  $\rank Q_j=\binom{D-1}{j-1}\rank Q_D$.
\end{statement}
The same result can be also proven using a decomposition of the
semisimple abelian group $G$ into a product of semisimple cyclic
groups, see the Appendix in Ref.~\onlinecite{Panteleev-Kalachev-2020}.

\subsection{Distance bounds.}
First we construct lower (existence) bounds on homological distances
in AMC complexes.  They are based on the results for $D$-dimensional
QHP codes\cite{Tillich-Zemor-2009,Zeng-Pryadko-2018,Zeng-Pryadko-hprod-2020},
and the fact that the AMC construction includes QHP codes derived from
abelian group-algebra codes as a special case.  We prove in Appendix
\ref{sec:proof-k-d-rates}:
\begin{statement}
  \label{th:k-d-rates}
  Let $G$ be a finite abelian group presented with the list of
  generators $\bs x=(x_1,x_2,\ldots,x_m)$, and
  $K\equiv G^{\times D}$ its $D$-fold tensor product with itself, with
  the corresponding list of generators
  $\{\bs x_1, \bs x_2,\ldots,\bs x_D\}$.  Given a
  checkable\cite{Borello-delaCruz-Willems-2022} abelian group-algebra
  code in $F[G]$ of length $n$, dimension $k$, and distance $d$,
  with the orthogonal element $a\equiv a(\bs x)$, consider the group
  algebra elements $a_i\equiv a(\bs x_i)\in F[K]$, $1\le i\le D$, and
  the corresponding $D$-complex ${\cal Q}=\amc(a_1,a_2,\ldots, a_D)$.
  Its $j$-th subspace has dimension
  $n_j\equiv \dim({\cal Q}_j)=\binom{D}{j}n^D$, while the corresponding
  homology and co-homology groups have the same dimension
  $k_j=\binom{D}{j}k^D$ and homological distances $d_j=d^j$ and
  $\tilde{d}_j=d^{D-j}$, respectively.
\end{statement}
In the special case of a symmetric $D=2j$ AMC code at level $j$,
this gives the encoding rate $k_j/n_j=(k/n)^D$ and the distance
$d^{D/2}$.  Thus, a finite rate and a linear distance scaling in the
original code family would give a finite rate and the distance scaling
as a square root of the code length, similar to GB codes equivalent to
QHP codes.  While we are not aware of a proof that good linear abelian
group-algebra codes exist (known results include good code families
based on \emph{non-abelian} dihedral groups $D_{2m}$\cite{Fan-Lin-2021}, random
quasi-abelian codes \cite{Bazzi-Mitter-2006}, and non-linear cyclic
codes\cite{Haviv-Langberg-Schwartz-Yaakobi-2017} that attain the
Gilbert-Varshamov bound), already the classical results on cyclic
codes\cite{Lin-Weldon-1967,Berlekamp-Justesen-1974} are sufficient to
show the existence of AMC codes with finite relative distances and
power-law distance scaling, with the exponent arbitrarily close to
$1/2$.   

Second, we give an upper bound on the distance for any AMC code based
on a cyclic group.  The bound, formulated in terms of the
characteristic polynomial (\ref{eq:characteristic-poly}) is a
generalization of a similar result for GB codes\cite{Wang-Pryadko-2022}.
\begin{statement}
  \label{th:distance-h-upper}
  Take a $D$-complex ${\cal Q}=\amc(\{a_i\}_{i=1}^D)$ over $F[C_\ell]$
  with characteristic polynomial $h(x)$.  Let $d_h^\perp$ be the
  minimal distance of the $F$-linear cyclic code with
  the check polynomial $h(x)$.  The homological distances satisfy the
  inequality $d_j({\cal Q})\le d_D({\cal Q})=d_h^\perp$, $1\le j< D$.
\end{statement}
When boundary operators from a level-$j$ homology group are used to
construct a CSS code, this gives a uniform upper bound on the code
distances, $d\le d_h^\perp$. We also get a bound on the single-shot
distance\cite{Campbell-2018} $d_{\rm SS}\le d_h^\perp$, a
parameter which may be relevant when the code is operated in a
fault-tolerant setting.

Notice, however, that the single-shot distance $d_{\rm SS}$ is only
relevant if data and syndrome decoding are done separately, and that
only if we are using unmodified matrices from the chain complex, e.g.,
$M_X=Q_{j-1}$ and $M_Z=Q_{j+2}^T$ with $H_X=Q_j$ and $H_Z=Q_{j+1}^T$.
In practice, the single-shot distance may be increased by adding
homologically non-trivial chains (co-chains) as additional rows of
$M_X$ ($M_Z$).  The downside is that these vectors may have higher
weights, potentially increasing the complexity of the syndrome
decoding step.

Another parameter relevant in a fault-tolerant setting is the
\emph{syndrome distance} $d_{\rm S}$, the minimum weight of a non-zero
syndrome vector.  An obvious upper bound is the minimum column weight
in the check matrix.  A boundary operator $Q_j$ in an AMC $D$-complex
${\cal Q}$ has $j$ or $D+1-j$ non-zero blocks in each block column or
row, respectively.  If we use group algebra elements of the same
weight $\omega$, this gives uniform column weights $j \omega$ and row weights
$(D+1-j)\omega$.  For a code $\css(Q_j,Q_{j+1}^T)$ in the space
${\cal Q}_j$ of such a $D$-complex, this gives a simple upper
bound on the syndrome distance,
\begin{equation}
  d_{\rm S}\le \min(j,D-j)\omega.\label{eq:syndrome-d-upper}
\end{equation}
We also get uniform weights of the stabilizer generators defined by
the rows of $H_X=Q_j$ and $H_Z=Q_{j+1}^T$,
\begin{equation}
  w_X=(D+1-j)\omega,\quad w_Z=(j+1)\omega.\label{eq:row-wX-wZ}
\end{equation}
In particular, with $D=4$ and $j=2$, this gives an upper bound for the
syndrome distance $d_{\rm S}\le 2\omega$ and stabilizer generators of a
uniform weight $w_X=w_Z=3\omega$.

\section{Numerical results}
\label{sec:numerical}
\subsection{Family of 4-dimensional rotated toric codes}

We now focus on the family of quasicyclic CSS codes
$\amc_2(a_1,a_2,a_3,a_4)$ constructed from distinct weight-$2$ group
algebra elements $a_i\in F[C_\ell]$, $1\le i\le 4$.  As explained in
the previous section, such codes are locally equivalent to
$4$-dimensional toric codes with qubits on faces (second homology
group).  While the topology of the code is the same as for the regular
4D toric codes, the lattice periodicity vectors $\bs\Delta_i$ need
no longer be parallel to the Cartesian axes.  Thus, we may 
refer to these codes as \emph{rotated 4D toric codes}.

Specifically, we constructed such codes for $7\le \ell\le 30$ and
corresponding block lengths $n=6\ell$ (four distinct weight-$2$
polynomials cannot be found for $\ell<7$).  Generating polynomials
and parameters of smallest codes for each minimum distance $d$ are
listed in Tab.~\ref{tab:toric-2222}.  All constructed codes have the
same dimension $k=6$ and the same characteristic polynomial
$h(x)=1+x$.  For comparison, in Tab.~\ref{tab:toric-2222} we also list
parameters of two conventional 4D toric codes, which require larger
block lengths to achieve the same distances.

In addition to code parameters and the defining polynomials,
Tab.~\ref{tab:toric-2222} shows the confinement
profiles\cite{Quintavalle-Vasmer-Roffe-Campbell-2021} of the
constructed codes, computed using the {\tt dist\_m4ri} program
\cite{Pryadko-2025-distm4ri}.  The confinement profile $w_{\rm S}(x)$
is defined as the minimum syndrome weight for an error $e$ with
irreducible weight $x=1,2,\ldots$.  Here the irreducible weight
$x=w(e)$ is the minimum weight of an error equivalent to $e$.  For any
non-zero $x<d$, the confinement profile is bounded from below by the
syndrome distance, $w_{\rm S}(x)\ge d_{\rm S}$, while
$w_{\rm S}(0)=w_{\rm S}(d)=0$, where $d$ is the code distance.

As discussed in Sec.~\ref{sec:equiv}, AMC codes constructed from
polynomials of weight two are locally equivalent to higher-dimensional
toric codes. They have identical confinement profiles, up to the
injectivity radius.  In particular, for a $4D$ toric code with qubits
on the plaquettes, codewords are two-dimensional membranes; creation
of a (partial) membrane of radius $r$ and area $\mathcal{O}(r^2)$
requires perimeter (syndrome weight) linear in $r$, which implies a
square root scaling for the confinement profile,
$w_{\rm S}=\mathcal{O}(w^{1/2})$, for error weights
$w\le w_{\rm max}=\mathcal{O}(L^2)$.  For such a code with block size
$n=6L^4$, the code distance is $d=L^2$, while the energy barrier
(maximum of the confinement profile) scales as
$\mathcal{O}(L)=\mathcal{O}\!\left(n^{1/4}\right)$.

For a large $4D$ toric code, the first five terms in the confinement
profile $w_{\rm S}(x)$ are $4,6,6,8,8$.  Only the first couple of
terms of the confinement profiles of the larger AMC codes in
Tab.~\ref{tab:toric-2222} match this sequence, consistent with the
fact that injectivity radii for these codes are small.  For example,
for the code with $\ell=30$, if we follow the prescription in
Sec.~\ref{sec:equiv} and define the generators $x_1=x^2$, $x_2=x^5$,
$x_3=x^8$, and $x_4=x^9$, we get a periodicity vector of length
$2 r_{\rm inj}=6^{1/2}<3$, as follows from the relation
$$x_1^2 x_2x_4^{-1}\equiv (x^2)^2x^5x^{-9}=1
$$ between thus defined generators.  Nevertheless, the
confinement profiles for the codes in Tab.~\ref{tab:toric-2222} are
higher than, e.g., for $3+3$ GB codes with syndrome distance
$d_{\rm S}=3$ previously studied in
Ref.~\onlinecite{Lin-Liu-Lim-Pryadko-circuits-2024}, where the
confinement profiles necessarily start with $3,4,3$, increasing only
for irreducible error weights $x\le 2$.

\begin{table}[htbp]
  \centering
  \begin{tabular}[c]{c||c|c|c|c||c|c|c|c||l}
    $\ell$&$n$&$k$&$d$&$d_{\rm S}$&$a_1$&$a_2$&$a_3$&$a_4$&confin
    \\ \hline
    7 & 42&6& 4 &4 &$1+x$&$1+x^2$& $1+x^3$ &$1+x^4$&4,4,4,6      \\
    10& 60&6& 5 &4 &$1+x$&$1+x^2$& $1+x^3$ &$1+x^4$&4,6,6,6,4     \\
    11& 66&6& 6 &4 &$1+x$&$1+x^2$& $1+x^3$ &$1+x^4$&4,6,6,6,4     \\
    14& 84&6& 7 &4 &$1+x$&$1+x^2$& $1+x^5$ &$1+x^6$&4,6,6,6,4     \\
    16& 96&6& 8 &4 &$1+x$&$1+x^3$& $1+x^5$ &$1+x^7$&4,6,8,8,4     \\
    18&108&6& 9 &4 &$1+x$&$1+x^3$& $1+x^5$ &$1+x^7$&4,6,8,8,4     \\
    25&150&6& 10&4 &$1+x$&$1+x^4$& $1+x^6$ &$1+x^9$&4,6,8,8,4     \\
    28&168&6& 11&4 &$1+x$&$1+x^3$&$1+x^7$&$1+x^{12}$& 4,6,8,8,4   \\
    30&180&6& 12&4&$1+x^2$&$1+x^5$& $1+x^8$ &$1+x^9$& 4,6,8,8,4   \\ \hline 
    $2^4$&96&6 &4&4&$1+x_1$&$1+x_2$&$1+x_3$&$1+x_4$ &4,4,4  \\
    $3^4$&486&6&9&4&$1+x_1$&$1+x_2$&$1+x_3$&$1+x_4$&4,6,6,8,4 
  \end{tabular}
  \caption{Smallest $D=4$ AMC$_2$ qubit codes (rotated 4D toric
    codes) with increasing distances $d$, constructed from cyclic
    groups $C_\ell$ with $\ell\le 30$ (above the line) and two
    conventional 4D toric codes (below the line).  Confinement
    profile \cite{Quintavalle-Vasmer-Roffe-Campbell-2021} (``confin'')
    column lists the minimum syndrome weights for irreducible errors
    of weights from $1$ to $5$, except for codes with $\ell=7$ and
    $\ell=2^4$ where the biggest irreducible errors have weights $4$
    and $3$, respectively.}
  \label{tab:toric-2222}
\end{table}

\subsection{Circuit design}
Stabilizer generator matrices of these codes, $H_X=R_2$ and
$H_Z=R_3^T$ are shown in Eqs.~(\ref{eq:mbc4-2}) and (\ref{eq:mbc4-3}),
where the blocks $A$, $B$, $C$, and $D$ should be replaced with the
corresponding circulant matrices.  With weight-two polynomials, $\omega=2$, this
gives stabilizer generators of uniform weight $w_X=w_Z=6$.  Each data qubit
(associated with a plaquette on a hypercubic lattice) is neighboring
with four $X$ checks (edges) and four $Z$ checks (cubes), total of $8$
checks.

We have set out to design quantum codes with increased redundancy of
low-weight stabilizer generators.  Evidently, constructed matrices
have \emph{too much redundancy} to measure all checks in a circuit of
optimal duration.  For this reason, in each round of measurement we
chose to address only the stabilizer generators corresponding to three
chosen row blocks, labeling the corresponding circuits by the index of
the dropped row block.

As an example, after removing the first and the last row blocks,
respectively, from the CSS matrices $H_X=R_2$ and $H_Z=R_3^T$ in
Eqs.~(\ref{eq:mbc4-2}) and (\ref{eq:mbc4-3}), we get in the case of a
binary field,
\begin{eqnarray}
  H_X^{(1)}&=&\left[
               \begin{array}[c]{ccc|ccc}
                 D_1&&&B_4&A_5\\ &D_2&&C_4&&A_6\\ &&D_3&&C_5&B_6
               \end{array}
                                                              \right],
  \label{eq:order-mat-Hx}\\
H_Z^{(1)}&=&\left[
  \begin{array}[c]{ccc|ccc}
    \widehat B_1&\widehat C_2&&\widehat D_4\\
    \widehat A_1&&\widehat C_3\,&&\widehat D_5\\
    &\widehat A_2&\widehat B_3\,&&&\widehat D_6\!
  \end{array}
\right],\;\,  \label{eq:order-mat-Hz}
\end{eqnarray}
where the subscripts next to block names label the columns and the
circumflex ``hat'' accent indicates transposed blocks.  Focusing on
bigger $3\times 3$ blocks, we have a two-block structure similar to
that in lifted-product or GB/2BGA codes, with two and four ("$2+4$")
non-zero entries in each subrow.  The circuits
designed\cite{Lin-Liu-Lim-Pryadko-circuits-2024} for $2+4$ 2BGA/GB
codes, a generalization of the \nz\ addressing
scheme\cite{Tomita-Svore-2014} for rotated surface codes, also work in
the present case.  Namely, a valid time-optimal circuit is obtained if
we address the qubits corresponding to the two terms in the $D$ blocks
at steps 1 and 6 (first and last), with the opposite order for $H_X$
and $H_Z$, while addressing the qubits from the remaining blocks in
the intervening time steps.  Specifically, we used the following
addressing scheme:
\begin{eqnarray*}
  \begin{array}[c]{l|ccc|ccc|}
    H_X&D_{10}&B_{40}&B_{41}&A_{50}&A_{51}&D_{11} \\
      &D_{20}&A_{60}&A_{61}&C_{40}&C_{41}&D_{21} \\
      &D_{30}&C_{50}&C_{51}&B_{60}&B_{61}&D_{31} \\[0.15em] \hline
    H_Z^{\strut}&\widehat D_{41}&\widehat B_{10}&\widehat B_{11}&\widehat C_{20}&\widehat C_{21}&\widehat D_{40}\\
      &\widehat D_{51}&\widehat C_{30}&\widehat C_{31}&\widehat A_{10}&\widehat A_{11}&\widehat D_{50} \\
      &\widehat D_{61}&\widehat A_{20}&\widehat A_{21}&\widehat B_{30}&\widehat B_{31}&\widehat D_{60}
  \end{array}
\end{eqnarray*}
Here each column corresponds to a circuit time step, and each row to a
block row of the matrices (\ref{eq:order-mat-Hx}),
(\ref{eq:order-mat-Hz}), with secondary subscripts labeling the
location corresponding to a free term (``0'') and the other monomial
(``1'').  For example, the 1st row indicates that for measuring
stabilizer generators defined by the first row block of
Eq.~(\ref{eq:order-mat-Hx}), one should first address the qubit
corresponding to the free term (second subscript 0) of the polynomial
defining the $D_1$ block, then the free term (second subscript 0) of
the $B_4$ block, then the other monomial (second subscript 1) of the
same block, etc.  Circuits corresponding to the original matrices with
2nd, 3rd, or 4th row removed were constructed identically, with the
help of column permutations rendering the matrices in the form
(\ref{eq:order-mat-Hx}), (\ref{eq:order-mat-Hz}), up to block labels.

Based on this addressing scheme for a single measurement round, we
designed three types of four-round syndrome measurement circuits, {\tt
  1111}, {\tt 1212}, and {\tt 1234}\footnote{To guarantee rank
  preservation after removal of a block row (which is important for
  the {\tt 1111} cycle), guided by Lemma~\ref{th:lemma-A-factor}, the
  polynomials were ordered so that $a_4=1+x$ [which gives block $D$ in
  Eqs.~(\ref{eq:mbc4-2}) and (\ref{eq:mbc4-3})].}.  Here a digit
stands for the row omitted in each particular round of measurement,
and we constructed detector events matching pairs of closest available
measurement outcomes, either from neighboring or
next-nearest-neighboring rounds.  Overall, all our circuits have fixed
duration of $T_{\rm max}=9$ measurement rounds, including two
four-round cycles and a final measurement of all data qubits in $Z$ or
$X$ basis, depending on the initial state of the data qubits,
$\ket 0^{\otimes n}$ or $\ket+^{\otimes n}$.  The results of the final
measurement are combined to construct the final set of detector events
(thus computed syndrome bits are exact), and also the logical
observables which are used to verify whether the decoding was successful.

\subsection{Circuit simulation results}
Thus designed measurement protocols have been implemented as {\tt
  Stim} circuits\cite{Gidney-2021-stim}, using {\tt CX} and {\tt XCX}
gates, with a standard circuit error model.  Namely, we included
probability-$p$ single-qubit depolarizing noise on the data qubits at
the start of each round, two-qubit depolarizing noise with probability
$p$ after every two-qubit gate, and also bit-flip (phase-flip) Pauli
error channel with the same probability $p$ before and after each
$Z$-basis ($X$-basis) measure-reset (MR) gate.  There are no
single-qubit gates in our circuits, and we chose to omit errors for
idling intervals due to the fact that in most platforms single-qubit
Clifford gates and idle intervals have error rates small compared to
those of two-qubit entangling gates or single-qubit MR operations.

For accurate decoding, one needs to know probabilities of individual
faults in the circuit and their effects on the detectors and the
logical observables.  This information for a given circuit is computed
automatically by {\tt Stim}, which also aggregates the faults
producing the identical effect, and can be exported as a detector
error model (DEM) file\cite{Gidney-2021-stim}.  Information about
aggregated faults in a DEM file can be regarded as columns of
stabilizer-generator $H_X$ and logical $L_X$ matrices (with known
aggregated fault probability for each column) for an asymmetric
quantum CSS code.  In particular, the distance $d_Z$ of this code is
identified as the circuit distance.  We used {\tt vecdec} program
\cite{Pryadko-2025-vecdec} to export the matrices from the DEM files
associated with individual circuits, and {\tt dist\_m4ri} program to
compute the corresponding distances (with the distances $d\le 6$
computed exactly with a greedy enumeration algorithm, and upper bounds
for larger distances using a random information set
algorithm\cite{Dumer-Kovalev-Pryadko-IEEE-2017}, exact with a high
probability).

Circuit simulations have been done with Pauli-frame circuit simulator
{\tt Stim} \cite{Gidney-2021-stim} written by Craig Gidney, and the
decoding using {\tt vecdec} program \cite{Pryadko-2025-vecdec}.
Specifically, before attempting belief-propagation (BP) decoding, the
program tries to match the syndrome data against a pre-computed list
of small-weight error clusters; such a simple pre-decoder gives a
substantial acceleration, especially at smaller error probabilities.
If pre-decoding fails, the corresponding clusters are discarded, and a
BP decoder with serial schedule, using both average and instantaneous
logarithmic likelihood ratios (whichever converges first) is engaged,
with up to 50 steps.  Finally, if that fails, the logarithmic
likelihood ratios (LLRs) are transmitted to the OSD-1 decoder.  A
sample of circuit simulation data for codes $[[42,6,4]]$, $[[66,6,6]]$,
and $[[96,6,8]]$ (see Tab.~\ref{tab:toric-2222}), with the standard
error model described earlier and the ``1212'' syndrome-addressing
cycle, is shown on a threshold-style plot in Fig.~\ref{fig:threshold}.
To speed up the decoding, we simplified the DEM by dropping all
``minority'' detector events, which somewhat increases logical error
rates\footnote{This is unlike with the {\tt BP+OSD}
  package\cite{roffe_decoding_2020,Roffe_LDPC_Python_tools_2022} whose
  performance may actually be degraded by additional detector events,
  see in Appendix A of Ref.~\onlinecite{Beni-Higgott-Shutty-2025}. In
  comparison, for the decoder used here, removing minority detector
  events increases logical error rates, e.g., from $0.20\%$ to
  $0.36\%$ at $p=0.1\%$ with the $[[42,6,4]]$ code and the ``1212''
  circuits.}.  Nevertheless, asymptotic slopes at small $p$ consistent
with $d/2$ (where $d$ is the code distance) are in agreement with the
circuit distances computed independently.  The (pseudo)threshold
position (crossing point) is at $p_c\approx 1.1\%$, in approximate
agreement with the circuit simulation results for similar codes in
Ref.~\onlinecite{Aasen-etal-Svore-2025}.  We note that our value is
(marginally) higher than the values for toric and surface codes under
the same error models.

\begin{figure}[htbp]
  \centering
  \includegraphics[width=\columnwidth]{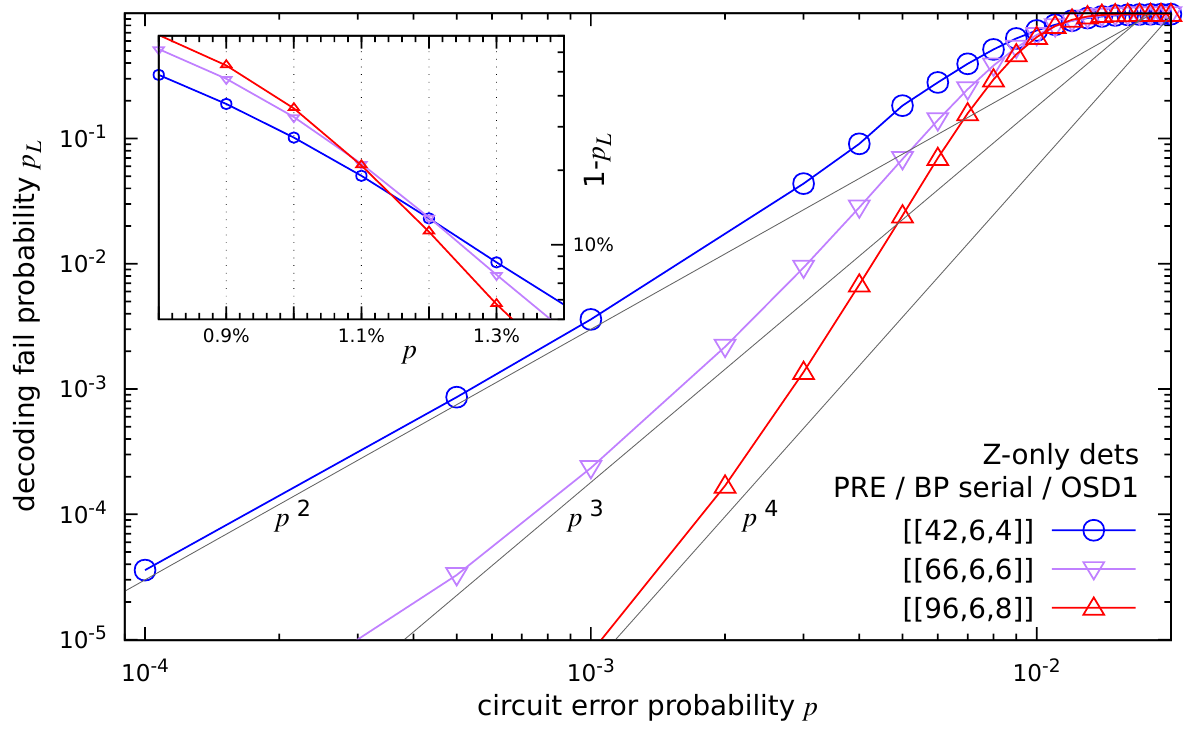}
  \caption{Logical error rate as a function of the circuit noise
    parameter $p$, using the syndrome-addressing scheme ``{\tt 1212}''
    and $Z$-basis final measurements.  For faster simulations, only
    $Z$-measurement detector events have been included.  Results for
    short even-distance codes $[[42,6,4]]$, $[66,6,6]]$, and
    $[[96,6,8]]$ (see Tab.~\ref{tab:toric-2222}) are shown.  Inset
    shows the success probability $1-p_L$ as a function of $p$ in the
    vicinity of the crossing point(s) at $p\approx 1.1\%$}
  \label{fig:threshold}
\end{figure}

In addition, to illustrate the ``single-shot'' properties of the
constructed codes, we have tried sequential one-step SW
decoding\cite{Breuckmann-Londe-2020,%
  Grospellier-Groues-Krishna-Leverrier-2020,Higgott-Breuckmann-2023},
with windows of different sizes.  Details of our implementation are
given in Ref.~\onlinecite{Lin-Liu-Lim-Pryadko-circuits-2024}.  Some
simulation results for the same three codes as in
Fig.~\ref{fig:threshold} are shown in Fig.~\ref{fig:single-shot},
where the logical error rates $p_L$ are plotted as a function of the
decoding window size $T$ at a fixed circuit error rate, $p=10^{-3}$.
Up to $5\times 10^6$ samples have been generated for each point,
resulting in very small statistical errors.

For the data in Fig.~\ref{fig:single-shot} we used full DEMs for
decoding, including the detector events associated with both $X$- and
$Z$-measurements, which resulted in somewhat lower $p_L$ at full-block
decoding ($T=9$) than those in Fig.~\ref{fig:threshold} at the same
$p$.  For each code, we tried all three syndrome measurement cycles.
Simplest is the ``{\tt 1111}'' cycle, where the same three row blocks
are measured in each round.  While discarding most of the redundant
generators, this cycle is the most accurate with full-block decoding.
In comparison, the ``{\tt 1234}'' cycle which addresses most uniformly
all low-weight stabilizer generators, gives the lowest fail rates with
single-shot ($T=1$) and two-shot ($T=2$) decoding protocols.  Note
that only during the rounds labeled with ``1'' a full-rank set of
stabilizer generators is measured.

\begin{figure}[htbp]
  \centering
  \includegraphics[width=\columnwidth]{amc_T_v2}
  \caption{Logical error rates computed with full detector error
    models (both $X$ and $Z$ detector events) for the same codes as in
    Fig.~\ref{fig:threshold}, plotted as a function of the decoding
    window size $T$.  Shading indicates the statistical error at one
    standard deviation.  Here $T=T_{\rm max}=9$, the total number of
    measurement rounds in our circuits, corresponds to (most accurate)
    full-block decoding, while $T<T_{\rm max}$ are the number of
    rounds used in SW decoding protocols.  The data labeled as ``{\tt
      1111}'', ``{\tt 1212}'', and ``{\tt 1234}'' correspond to
    different syndrome measurement cycles as explained in the text.
    At larger $T$, syndrome measurement errors are less important, and
    the ``{\tt 1111}'' cycle gives lower error rates due to shorter
    measurement cycle.  In contrast, with single- ($T=1$) and two-shot
    ($T=2$) decoding, the error rates are lower with the other two
    cycles which better utilize the syndrome redundancy.  In
    comparison, in the ``1111'' measurement cycle, most of the
    redundant checks are discarded.  The difference is substantial for
    the code $[[96,6,8]]$ with a larger minimal distance $d=8$, see a
    discussion in Sec.~\ref{sec:disc-sim}.}
  \label{fig:single-shot}
\end{figure}

\subsection{Discussion of the simulation results}
\label{sec:disc-sim}

First, we note that LDPC codes, quantum or classical, have Tanner
graphs of bounded degree.  Consequently, a high-weight stochastic
error—even one exceeding the code distance—can be decomposed into a
collection of smaller, localized clusters. These clusters can then be
decoded independently \cite{Kovalev-Pryadko-FT-2013,
  Dumer-Kovalev-Pryadko-bnd-2015}. Fault-tolerant circuit detector
error models (DEMs) for LDPC codes inherit this same property of
locality.

For such a code or circuit with distance $d$ and a sufficiently small
error probability $p$, the logical error probability under
minimum-weight decoding scales as
$$
P_{\rm fail}=AN\,{\textstyle\binom{d}{t+1}}\,p^{t+1}.
$$
Here $A$ is an order-one prefactor, $t \equiv \lfloor (d-1)/2 \rfloor$
is the number of corrected errors, and the binomial factor represents
the number of failing configurations within a codeword.  $N$ denotes
the number of irreducible codewords where a fault of $t+1$ errors can
cause a decoding failure.  Specifically: when $d=2t+2$ (even), $N$ is
the number of minimum-weight codewords; when $d=2t+1$ (odd), $N$ can
be estimated as the number of codewords with weights $d$ and
$d+1$. For precise definitions of these terms, see
Ref.~\onlinecite{Dumer-Kovalev-Pryadko-bnd-2015}.

A sliding window decoder is inherently suboptimal compared to a
full-circuit decoder, and its decoding accuracy worsens with the
decreasing window size $T$.  Depending on the code at hand, such a
decoder may fail for certain errors of weight $t'+1<t+1$, leading to
an asymptotically higher logical error rate
$P_{\rm fail}=\mathcal{O}(p^{t'+1})$.  Alternatively, it may fail only
on additional error configurations of the same weight $t+1$; this
results in a similar scaling at small $p$ but with an increased
prefactor $N'>N$. For surface codes, a window size of $T\ge d$ is
sufficient to achieve optimal decoding accuracy\cite{Stephens-2014,
  Das-Locharia-Jones-2022}.

One mechanism that can cause a dramatic increase in the prefactor $N$
(at fixed $t$) for small window sizes is the LDPC-breakdown transition
\cite{Lin-Liu-Lim-Pryadko-circuits-2024}. Consider an error $e$ with
irreducible weight $w\equiv \wgt(e)$ and syndrome weight $w_s$ (the
number of detectors flipped in a single measurement round). A one-step
sliding window (SW) decoder with window size $T$, using minimum-weight
decoding, will misinterpret the entire error $e$ as a sequence of
measurement errors if
\begin{equation}
  \label{eq:LDPC-breakdown}
  w_s T <w.
\end{equation}
Such an error is perpetuated until the final round of measurement
where it is corrected.  While this effect alone does not reduce the
circuit distance, it creates a vulnerability.  In codes with a flat
energy profile, LDPC breakdown allows a chain of additional,
infrequent weight-one errors to eventually trigger a decoding
failure. This increase in the number of failing error configurations
is entirely entropic. In contrast, with larger $T$, a decoding failure
typically requires a weight-$(t+1)$ error to be localized within a
single measurement cycle. Notably, a high energy barrier or a high
confinement profile suppresses this process, as additional data errors
that increase $w$ may eventually violate condition
(\ref{eq:LDPC-breakdown}), allowing the decoder to correctly identify
and resolve the error.

SW decoding with a smaller $T$ can also be viewed as a variant of
local decoding.  The most well-understood analogue is
finite-temperature self-correction in a Hamiltonian model
corresponding to a code, which is quantified by the code's {energy
  profile} and {energy
  barrier}\cite{Bombin-Chhajlany-Horodecki-MartinDelgado-2013}.  For a
surface code the energy profile is flat, $w_s(x)=1$.  Consequently,
once a weight-one error is created, additional errors can be added
without penalty, effectively diffusing the endpoint of a string-like
fault through the system and rendering such codes unstable at finite
temperatures.  In contrast, 4D toric codes exhibit self-correction.
Their energy profile is lower-bounded by a confinement profile scaling
as $w_{\rm S}(x)=\mathcal{O}(x^{1/2})$, with an energy barrier scaling
as $\mathcal{O}(L)=\mathcal{O}\!\left(n^{1/4}\right)$.  As a result,
local decoders are efficient for such codes\cite{Hastings-hyp-2013}.
For 4D AMC codes---such as those constructed from weight-$2$
polynomials---we expect properties similar to 4D toric codes (up to
the injectivity radius) when all minimum-weight stabilizer generators
are measured. Furthermore, single-shot decoding and the associated
sustained thresholds can be analyzed using the techniques proposed in
Ref.~\onlinecite{Quintavalle-Vasmer-Roffe-Campbell-2021}.

In practice, we do not have a time-optimal scheme to measure the
entire set of $4\ell$ minimum-weight stabilizer generators at once.
Instead, each measurement round only accesses generators in three row
blocks of CSS matrices.  When we address the same three row blocks in
every measurement round (the data labeled ``{\tt 1111}'' in
Fig.~\ref{fig:single-shot}), we get a code with CSS matrices
(\ref{eq:order-mat-Hx}), (\ref{eq:order-mat-Hz}), with the same energy
profile as for the toric code, $w_{\rm S}(x)=2$, due to the overall
parity in the $D$ blocks.  Conversely, for the ``{\tt 1212}'' and
``{\tt 1234}'' circuits, some of the measurements are only verified
after skipping a cycle.  For these circuits, we expect the confinement
profiles in Tab.~\ref{tab:toric-2222} to become fully relevant only
for SW decoding with $T\ge2$.

How do we interpret the results in Fig.~\ref{fig:single-shot}?  For
the circuits used, we have verified that our SW decoder with $T\ge2$
preserves the code distance, i.e., it corrects any $t=1$, $2$, and $3$
circuit errors for the three codes with distances $d=4$, $6$, and $8$,
respectively.  Evidently, for all $d=4$ circuits, SW decoding with
$T=2$ is sufficient to reach nearly the same accuracy as for $T=9$
full-block decoding.  For the $d=6$ circuits, $T=3$ gives results
nearly identical to $T=9$, with slightly lower logical error rates
observed for the shorter measurement cycle labeled ``{\tt 1111}''.
The decoding failure probability increases as window size is decreased
to $T=2$, and here it is highest for the ``{\tt 1111}'' cycle---a
result consistent with the expectation that the other two cycles offer
better stability against measurement errors.  For the $d=8$ circuits,
we see an even more pronounced increase of the logical error rates for
$T=2$ compared to $T=3$ SW decoding (under the ``{\tt 1111}'' cycle),
as well as compared to the $T=2$ data for the other two cycles.

We notice that for the ``{\tt 1111}'' circuit, characterized by a flat
effective confinement profile $w_{\rm S}(x)=2$, for a window size
$T=2$, Eq.~(\ref{eq:LDPC-breakdown}) predicts a potential
LDPC-breakdown transition only for errors of weight $w>4$.  A tie
occurs at $w=4$, which may be broken by our decoder if those specific
syndrome errors are statistically more likely.  It is apparent from
Fig.~\ref{fig:single-shot} that this effect is sufficient to increase
the prefactor $N$ by a factor of two or three.

Clearly, while one- or two-shot decoding may function for the codes
examined, it may not yield optimal error rates as the code distance
increases.  Larger decoding window sizes $T(x)$ may be necessary to
ensure that the LDPC-breakdown transition does not occur for errors of
weight $x$.  For a power-law confinement profile
$w_{\rm S}(x)=\mathcal{O}( x^\alpha)$, Eq.~(\ref{eq:LDPC-breakdown})
indicates that $T(x)=\mathcal{O}(x^{1-\alpha})$ is sufficient to
prevent the breakdown transition.  For larger 4D toric codes and
locally equivalent AMC codes with the confinement exponent
$\alpha=1/2$, choosing $T=\mathcal{O}(d^{1/2})$ should be sufficient
to fully suppress LDPC breakdown.

This argument also implies that linear confinement, $\alpha=1$, would be
sufficient to ensure that large codes within a family achieve their
design error rates with a constant window size, $T={\cal O}(1)$.  This
result aligns with the claim that linear confinement guarantees
single-shot fault-tolerant decoding in the presence of stochastic
errors\cite{Quintavalle-Vasmer-Roffe-Campbell-2021}.

\section{Conclusion}

In conclusion, we constructed a family of quantum error-correcting
codes whose main advantage is a highly redundant set of minimum-weight
stabilizer generators, with relatively short block lengths.
Just like the higher-dimensional QHP codes which our codes generalize,
the construction defines a sequence of chain complexes of varying
lengths.  The constructed $D$-dimensional codes are highly symmetric
and they are characterized by confinement profiles and single-shot
properties similar to those of the $D$-dimensional QHP codes.

It is instructive to compare the parameters of $D$-dimensional AMC
codes over $F[C_\ell]$ with a characteristic polynomial $h(x)$, and
$D$-dimensional QHP codes constructed from a linear cyclic code with
parameters $[\ell,\kappa,\delta]$ and the same check polynomial
$h(x)$.  At level $j$, the QHP code has block length
$N=\binom{D}{j}\ell^D$, dimension $K=\binom{D}{j}\kappa^D$, and
minimal distance $\min(\delta^j,\delta^{D-j})$, i.e., its rate
scales as a power of that of the original code, $R=(\kappa/\ell)^D$.  In
comparison, the block lengths and the dimensions of AMC codes scale
linearly with $\ell$, $n=\binom{D}{j}\ell$ and $k=\binom{D}{j}\kappa$,
so that the rate does not change, which is also a feature of GB codes.

While there is no explicit expression for the distance, we can use the
codes in Tab.~\ref{tab:toric-2222} to compare the actual distances
with the corresponding upper bounds from Statement
\ref{th:distance-h-upper}, which for these codes are given by the
block size, $\ell$.  For $d\le 9$, we have a pretty good linear
scaling, $d\ge \ell/2$, while for the remaining three codes,
$d>\ell/3$.  (Unfortunately, we do not expect the distances to remain
as high for larger codes.)

The most important property of the proposed construction is the high
redundancy of its minimum-weight stabilizer generators.  For a
quasi-abelian two-block code of dimension $k=2\kappa$, there are
exactly $\kappa$ minimum-weight redundant generators of each type ($X$
and $Z$). In contrast, the stabilizer generator matrices of a
symmetric AMC code with even $D=2j$ have ranks
$\binom{2j-1}{j-1}(\ell-\kappa)$ and possess $\binom{2j}{j-1}\ell$
rows.  This provides a redundancy of over $25\%$ for $D=4$ and over
$33\%$ for $D=6$.  Such redundancy can be leveraged to accelerate
decoding and improve accuracy under realistic circuit noise.  This is
illustrated by our simulations of relatively short AMC codes, which
demonstrate two-shot decoding comparable in accuracy to full-block
decoding.  We also find a pseudothreshold above $1\%$---exceeding
those of surface codes (with weight-$4$ generators) and GB codes with
weight-$6$ generators.  Further studies are required to establish the
long-term performance of the constructed codes; in particular, how
sustained error rates depend on the decoding window size $T$.

\section*{Data availability}

The data that support the findings of this study are available from
the authors upon reasonable request.

\begin{acknowledgments} This work was supported in part by the APS
  M. Hildred Blewett Fellowship (HKL), and the NSF awards 2112848
  (LPP) and OIA-2044049 (AAK).
\end{acknowledgments}

\appendix
\section{Proof of Statement \ref{th:k-cyclic}}
\label{sec:proof-k-cyclic}

The calculation is based on the fact that the group algebra
$F[C_\ell]$, isomorphic to the ring $R\equiv F[x]/(x^\ell-1)$ of modular
polynomials, is a gcd domain.  That is, for any two polynomials
$a(x),c(x)\in F[x]/(x^\ell-1)$, there exists a polynomial
$h_{ac}(x)\equiv\gcd(a,c)\in R$, the greatest common divisor of the
two original polynomials.  Moreover, there also exist B\'ezout
coefficients $u(x),v(x)\in F[x]/(x^\ell-1)$ such that
\begin{equation}
  u(x)a(x)+v(x)c(x)=h_{ac}(x)\equiv \gcd(a,c).
  \label{eq:Bezout}
\end{equation}
Equivalently, if we introduce the rescaled polynomials
$a'(x)=a(x)/h_{ac}(x)$, $c'(x)=c(x)/h_{ac}(x)$, this gives
$u(x)a'(x)+v(x)c'(x)=1$.  Rescaled polynomials and the B\'ezout
coefficients can be used to construct unimodular (and thus invertible) matrices,
$$
M_{\pm}(x)=\left(
  \begin{array}[c]{cc}
    a'(x)& \pm c'(x)\\ \mp v(x)&u(x)
  \end{array}
\right),\quad \det M_{\pm}=1.
$$
The corresponding two-block matrices are also invertible,
$$
\hat W_{\pm}=\left(
  \begin{array}[c]{cc}
    A'& \pm C'\\ \mp V&U
  \end{array}
\right),\quad \hat W_{\pm}^{-1} =\left(
  \begin{array}[c]{cc}
    U& \mp C'\\ \pm V&A'
  \end{array}
\right),
$$
where $A'=\MM(a')$, $C'=\MM(c')$, etc.\ are $\ell\times \ell$
circulant matrices constructed from the corresponding polynomials; we
have $AC'-CA'=0$ and $UA+VC=H_{ac}\equiv \MM(h_{ac})$.

Our approach is to construct a sequence of row and column
transformations diagonalizing the boundary operator matrices in an AMC
complex.  First, demonstrate the steps on the first two matrices in
the 3-complex ${\cal R}=\amc(a,b,c)$, see Eqs.~(\ref{eq:mbc3-1}) and
(\ref{eq:mbc3-2}).  Here, we want to eliminate matrices $C=\MM(c)$,
while replacing all matrices $A$ with $H_{ac}=\MM(h_{ac})$.  To this
end, denote $W^{\pm}_{i,j}$ the unimodular block-diagonal matrices
constructed from an identity matrix by replacing blocks in positions
$\{(i,i), (i,j), (j,i), (j,j)\}$ with the blocks of the corresponding
matrices $\hat W_{\pm}$ (row-first order), $\overline{W}^{\pm}_{i,j}$
the corresponding inverse matrices, and $W^{\pm}_{i,j}$,
$\overline{W}^{\pm}_{i,j}$ the same matrices with additional
sign-flip.  Explicitly,
\begin{eqnarray*}
  Q_1'&=&Q_1\overline{W}^+_{31}=[C,B,A] \left[
          \begin{array}[c]{ccc}
            A'&&V\\&I\\ -C'&&U
          \end{array}\right]  =[0,B,H_{ac}],\\
  Q_2'&=&{W}^+_{3,1}Q_2 \overline{W}^-_{3,1}\\
  &=&\left[
          \begin{array}[c]{ccc}
            +U&&-V\\&I\\ +C'&&+A'
          \end{array}\right]\left[
  \begin{array}[c]{ccc}
    B&A&0\\ -C&0 &A\\ 0& -C& -B
  \end{array}
                             \right] \overline{W}^-_{3,1}\\
      &=& \left[
  \begin{array}[c]{cc|c}
   UB&H_{ac}&VB\\ \hline -C&0 &A\\ C'B& 0&- A'B
  \end{array}
  \right] \left[
          \begin{array}[c]{ccc}
            A'&&-V\\&I\\ C'&&U
          \end{array}\right]\\
  &=&\left[
  \begin{array}[c]{cc|c}
   B&H_{ac}&0\\ \hline 0&0 &H_{ac}\\ 0& 0& -B
  \end{array}
  \right] .
\end{eqnarray*}
The matrices are now reduced to a block-diagonal form, with the blocks
being the matrices from the 2-dimensional complex $\amc(h_{ac},b)$,
which, with the help of the identity
$$\gcd(h_{ac},b,x^\ell-1)=\gcd(a,b,c,x^\ell-1)\equiv h(x),$$
immediately gives
\begin{eqnarray*}
  \rank Q_1&=&\rank Q_1'=\ell-\deg h=\ell-\kappa,\\
  \rank Q_2&=&\rank Q_2'=2(\ell-\deg h)= 2(\ell-\kappa).
\end{eqnarray*}
The elimination process goes the same way for larger-$D$ complexes.
Due to symmetry, whenever a pair of matrices, say, $A$ and $B$, are
located in positions $(x_1,y_1)$, $(x_1,y_2)$ (same column) and in
$(x_2,y_3)$, $(x_3,y_3)$ (same row), the corresponding $3\times3$
block (rows $\{y_1,y_2,y_3\}$, columns $\{x_1,x_2,x_3\}$) will
necessarily have the form of $Q_2$ in Eq.~(\ref{eq:mbc3-2}), with only
one additional block (e.g., $C$), up to some signs and the block
labels.

We illustrate the elimination process on the example of the first two
matrices from the 4-complex $\amc(a,b,c,d)$ in Eqs.~(\ref{eq:mbc4-1}),
(\ref{eq:mbc4-2}).  Here $ua+vd=h_{ad}\equiv \gcd(a,d,x^\ell-1)$,
$a'=a/h_{ad}$, $d'=d/h_{ad}$, and capital letters correspond to block
matrices, $A'=\MM(a')$, $U=\MM(u)$, etc.
\begin{widetext}
{\small
\begin{eqnarray*}
  R_1'&=&R_1\overline{W}^+_{4,1}=[D,C,B,A]
          \left[
          \begin{array}[c]{cccc}
            A'&&&V\\&I\\ &&I\\-D'&&&U
          \end{array}\right]  =[0,C,B,H_1], \text{ and}\\
      R_2'&=&W^-_{4,1}R_2\overline{W}^-_{5,1}\overline{W}^-_{6,2}\\
      &=&          \left[
          \begin{array}[c]{cccc}
            U&&&-V\\&I\\ &&I\\D'&&&A'
          \end{array}\right]
\left[
         \begin{array}[c]{ccc|ccc}
           C& B& A \\ \hline
           -D&&&B& A&0\\&-D&& -C&0&A\\ &&-D&0&-C&-B\\
         \end{array}
  \right]\overline{W}^-_{5,1}\overline{W}^-_{6,2}\\
      &=&
          \left[
         \begin{array}[c]{ccc|ccc}
           UC& UB& H_1&0&VC&VB \\ \hline
           -D&&&B& A&0\\&-D&& -C&0&A\\D'C &D'B&0&0&-A'C&-A'B\\
         \end{array}
  \right]\left[
  \begin{array}[c]{cccccc}
    A'&&&&-V\\&I\\&&I\\&&&I\\D'&&&&U\\&&&&&I
  \end{array}
                                            \right] \overline{W}_{6,2}^{-}\\
      &=&
                    \left[
         \begin{array}[c]{ccc|ccc}
           C& UB& H_1&0&0&VB \\ \hline
           0&&&B& H_1&0\\
            &-D&& -C&0&A\\
            &D'B&0&0&-C&-A'B\\
         \end{array}
  \right]\left[
  \begin{array}[c]{cccccc}
  I\\  &-V&&&&A'\\&&I\\&&&I\\&&&&I\\&U&&&&D'
  \end{array}\right]=\left[
         \begin{array}[c]{ccc|ccc}
           C& B& H_1&&& \\ \hline
           0&&&B& H_1&0\\
            &0&& -C&0&H_1\\
            &&0&0&-C&-B\\
         \end{array}
  \right].
\end{eqnarray*} }
\end{widetext}
The proof can be now completed recursively.  Indeed, the rank
expression in Statement \ref{th:k-cyclic} is valid for $D=2$ which
corresponds to a GB code based on a cyclic group.  Assuming it is
valid for some $D\ge 2$, take a circulant matrix $N=\MM(n(x))$ and use
Eqs.~(\ref{eq:mbcD-1}), (\ref{eq:mbcD-2}), and (\ref{eq:mbcD-3}) to
construct the boundary operators for the dimension $D+1$.  Now, use
the reduction protocol described to eliminate matrices $N$ and replace
$A\to \MM(h_{a,n})$, where $h_{a,n}=\gcd(a,n,x^\ell-1)$.  By
assumption, the ranks at level $j$ are
$\rank Q_j=\binom{D-1}{j-1}(\ell-\kappa)$, where $\kappa=\deg h(x)$, and the
expression for the characteristic polynomial now includes all $D+1$
polynomials including $n(x)$.  This gives
\begin{eqnarray*}
  \rank R_j^{D+1}&=&\binom{D-1}{j-2}(\ell-\kappa)+\binom{D-1}{j-1}(\ell-\kappa)\\
  &=&\binom{D}{j-1}(\ell-\kappa),
\end{eqnarray*}
in agreement with the assumption.  This completes the
induction.\hfill\qed

\section{Proof of Statement \ref{th:k-semisimple}}
\label{sec:proof-k-semisimple}

The proof is based on the following:
\begin{lemma}
  \label{th:lemma-A-factor}
  For a $D$-complex ${\cal Q}=\mbc(A,B,\ldots)$, suppose one of the matrices,
  e.g., matrix $A$, which does not limit generality, can be factored
  out from the other matrices, that is, $B=B'A=AB'$, $C=C'A=AC'$,
  etc., then $\rank Q_j^{(D)}=\binom{D-1}{j-1}\rank A$.
\end{lemma}
\begin{proof}
  Permute the terms in the product so that $A$ enters last, and use
  the equations (\ref{eq:mbcD-1}), (\ref{eq:mbcD-2}), and
  (\ref{eq:mbcD-3}).  The block $N=A$ enters $Q_j$ only along the
  diagonal of the square block of block-dimension $\binom{D-1}{j-1}$;
  the remaining matrices are above and to the right.  Using the
  assumed factorization, $B=B'A$, etc., row and column transformation
  can be used to diagonalize the matrix $Q_j$, leaving only the square
  block with $A$ along the diagonal.  This gives
  $\rank Q_j=\binom{D-1}{j-1}\rank A$.
\end{proof}

\noindent{\bf Proof of Statement \ref{th:k-semisimple}.} By
assumption, the group algebra $F[G]$ is semisimple; every ideal in
such a group algebra is generated by an idempotent element.  That is,
for any $a\in F[G]$, there is an idempotent element $e_a\in F[G]$
(necessarily $e_a=e_a^2$) such that $e_a a=a$ and, necessarily,
$(1-e_a)a=0$.  Then, given $\{a,b,c,\ldots\}\subset F[G]$, construct
the corresponding matrices $A=\MM(a)$, $B=\MM(b)$, etc, as well as the
idempotent matrices $E_a=\MM(e_a)$, $E_b=\MM(e_b)$, etc; all of these
matrices commute and satisfy, e.g., $E_AA=AE_A=A$, $(I-E_A)A=0$.  For
a $D$-complex $\amc(a,b,c,\ldots)$, there are $D$ group algebra
elements and $D$ corresponding idempotent matrices.

Further, by the same assumption, for any $a\in F[G]$, there exists a
pseudoinverse $x\equiv \underline a\in F[G]$ such that
$ax = e_a ax =e_a$.  Denote the corresponding matrix
$\underline A=\MM(\underline a)$, with the property
$A \underline A=E_A$.  Together with $A=E_AA$, this implies
$$
\rank ( E_{A}) = \rank (A).
$$

Given a binary string
$\bs \alpha=(\alpha_1,\alpha_2,\ldots,\alpha_D)\in \mathbb{F}_2^D$,
consider the idempotent matrices
$$
E_{\bs\alpha}\equiv \prod_{i=1}^D E_i^{\alpha_i}(I-E_i)^{1-\alpha_i}.
$$
By construction, these commute with all block matrices $A$, $B$,
\ldots, and satisfy the following orthogonality and completeness
properties,
\begin{equation}
\label{eq:completeness}
E_{\bs \alpha} E_{\bs \alpha'}=I_\ell\delta_{\bs\alpha,\bs \alpha'},\quad
\sum_{\bs \alpha\in \mathbb{F}_2^D} E_{\bs \alpha}=I_\ell.
\end{equation}
Consider a boundary operator (matrix) $Q_j$, $1\le j\le D$ of a known
block-size $r\times c$, and the corresponding matrix $Q_j(\bs \alpha)$
with each block row multiplied by $E_{\bs\alpha}$.  By properties of
idempotent matrices $E_{\bs\alpha}$, non-zero rows in
$Q_j(\bs \alpha)$ with different $\bs \alpha$ are linearly
independent.  Also, this is an invertible transformation, thus
$$\rank Q_j=\sum_{\bs \alpha\in \mathbb{F}_2^D}\rank Q_j(\bs \alpha).$$

Now, consider $Q_j(\bs \alpha)$ for $\bs \alpha\neq0$; without
limiting generality, assume $\alpha_1\neq0$, so that the matrix block
$E_{\bs \alpha}A\neq0$ has the same rank as the corresponding
idempotent $E_{\bs \alpha}$.  Due to the existence of the
pseudoinverse $\underline A$, for every other matrix block, e.g.,
originating from the matrix $B$, we have
$$
E_{\bs \alpha} B = E_{\bs \alpha} E_AA\underline AB=(E_{\bs \alpha}A)\,\underline AB.
$$
Evidently, thus projected blocks contain factors $E_{\bs \alpha} A$
and thus satisfy the conditions of the Lemma \ref{th:lemma-A-factor},
which gives
$$
\rank Q_j(\bs \alpha)=\binom{D-1}{j-1}\rank E_{\bs\alpha} A=\binom{D-1}{j-1}\rank E_{\bs\alpha},
$$
where $\alpha_1\neq 0$. The same construction works with other
non-zero $\bs \alpha$, while for $\bs \alpha=\bs 0$ the projected
matrix is identically zero, and thus $\rank Q_j(\bs0)=0$.  This gives,
finally, with the help of the completeness relation in
Eq.~(\ref{eq:completeness}):
\begin{eqnarray*}
  \rank Q_j&=&\sum_{\bs\alpha\in\mathbb{F}_2^D}\rank Q_j(\bs \alpha)\\
           &=&\sum_{\bs \alpha}(1-\delta_{\bs 0,\bs\alpha})\binom{D-1}{j-1}\rank E_{\bs \alpha}\\
           &=&\binom{D-1}{j-1}[\rank I_\ell -\rank E_{\bs 0}]\\
           &=&\binom{D-1}{j-1}(\ell -\kappa),\quad  \kappa\equiv\rank E_{\bs 0},
\end{eqnarray*}
with $\kappa$ being the rank of the products of idempotent matrices
$(I-E_A)$, $(I-E_B)$, \ldots, complementary to the commuting blocks
$A$, $B$, \ldots used to construct the $D$-complex.

The proof is completed by noticing that for $j=D$, the binomial
coefficient equals $1$, thus $\rank Q_D=\ell-\kappa$.\qed

\section{Proof of Statement \ref{th:k-d-rates}}
\label{sec:proof-k-d-rates}

The proof amounts to an observation that the construction gives a
$D$-dimensional QHP code constructed from the $F$-linear abelian group
code with the check matrix $A=\MM(a)$.  With $k=0$, all homology
groups are trivial, and all distances are infinite.  Otherwise, the
$1$-complex with the boundary operator $A$ has both subspaces of
dimensions $\ell$, both homology groups of dimensions $k$, and the
homological distances $d_0^{(1)}=1$, $d_1^{(1)}=d$. This gives the
parameters of the product complex: homology group dimensions $k_j$
follow immediately the K\"unneth formula (\ref{eq:tp-Kunneth}) as the
coefficient of $x^j$ in the Poincar\'e polynomial $(k+xk)^D$, subspace
dimensions similarly from Eq.~(\ref{eq:tp-nj}) written as a Poincar\'e
polynomial $(\ell+x\ell)^D$, and distances from Theorem 17 in
Ref.~\onlinecite{Zeng-Pryadko-hprod-2020}.

\section{Proof of Statement \ref{th:distance-h-upper}}
\label{sec:distance-h-upper}

It is trivial to verify that the Statement is true for $j=0$ and for
$j=D$.  Denote $r\equiv m_j^D=\binom{D}{j}$ and
$c\equiv m_{j+1}^D=\binom{D}{j+1}$ the block dimensions of subspaces
${\cal Q}_j$ and ${\cal Q}_{j+1}$, respectively.  For any non-zero
vector $\bs u$ in $C_h^\perp$, we can construct $r$ same-weight
vectors $\bs u_i$ in the space ${\cal Q}_j$ (with non-zero elements in
only one block, and that equal to $\bs u$) which satisfy the equation
$Q_{j}\bs u_i=\bs 0$.  We would like to show that all of them cannot
be homologically trivial at the same time.

To this end, define $g(x)=(x^\ell-1)/h(x)$ and write the non-zero
codeword as $u(x)=i(x)g(x)$, with $\deg i(x)<k$, and assume all $r$
vectors $\bs u_i$ to be trivial.  The corresponding condition can be
written as a matrix-polynomial equation
$$
u(x) I_{r\times r}=Q_j(x) \Xi_{r\times c}(x) \bmod x^\ell-1,
$$
where $\Xi_{r\times c}(x)$ is a matrix of unknowns of specified
dimensions, with elements in the ring $F[x]/(x^\ell-1)$.  With the
help of the invertible transformations similar to those constructed in
Sec.~\ref{sec:proof-k-cyclic}, matrix $Q_j(x)$ can be transformed to a
diagonal form,
$$
Q_j(x)=W_L(x) h(x)\Lambda_{r\times c} W_R(x),
$$ with square matrices
$W_L(x)$ and $W_R(x)$ invertible in $F[x]/(x^\ell-1)$, and
$\Lambda_{r\times c}$ with only
$$ \rank \Lambda_{r\times c}= m_j^{D-1}<\min(r,c)
$$ non-zero elements
(equal to 1), all in different rows and columns.  Using the fact that
the object on the left-hand-side (l.h.s.)\ is proportional to the
identity matrix, after a multiplication by $W_L^{-1}(x)$ on the left
and $W_L(x)$ on the right, we have
\begin{eqnarray*}
u(x)I_{r\times r}&=&W_L(x) h(x) \Lambda_{r\times c} W_R(x) \,\Xi_{r\times c}(x) \bmod x^\ell-1,\\
  u(x)I_{r\times r}&=& h(x) \Lambda_{r\times c} {\Xi'_{r\times c}(x)} \bmod x^\ell-1,
\end{eqnarray*}
where the arbitrary unknown coefficients in $\Xi(x)$ have been
replaced with the unknown coefficients in
$\Xi'(x)\equiv W_R(x) \,\Xi_{r\times c}(x)\,W_L(x)$; this
transformation is one-to-one since the transformation matrices are
invertible by construction.

If we pick a zero column in $\Lambda_{r\times c}$, the equation reads
$u(x)=0\bmod x^\ell-1$, contradicting the initial assumption.  Thus,
for any non-zero codeword in $C_h^\perp$, there is a non-trivial
vector in $H_j(\mathcal{Q})$ of the same weight, which proves the
stated upper bound on the distance.\qed

\bibliography{lpp,qc_all,more_qc,ldpc,linalg}

\begin{thebibliography}{70}%
\makeatletter
\providecommand \@ifxundefined [1]{%
 \@ifx{#1\undefined}
}%
\providecommand \@ifnum [1]{%
 \ifnum #1\expandafter \@firstoftwo
 \else \expandafter \@secondoftwo
 \fi
}%
\providecommand \@ifx [1]{%
 \ifx #1\expandafter \@firstoftwo
 \else \expandafter \@secondoftwo
 \fi
}%
\providecommand \natexlab [1]{#1}%
\providecommand \enquote  [1]{``#1''}%
\providecommand \bibnamefont  [1]{#1}%
\providecommand \bibfnamefont [1]{#1}%
\providecommand \citenamefont [1]{#1}%
\providecommand \href@noop [0]{\@secondoftwo}%
\providecommand \href [0]{\begingroup \@sanitize@url \@href}%
\providecommand \@href[1]{\@@startlink{#1}\@@href}%
\providecommand \@@href[1]{\endgroup#1\@@endlink}%
\providecommand \@sanitize@url [0]{\catcode `\\12\catcode `\$12\catcode
  `\&12\catcode `\#12\catcode `\^12\catcode `\_12\catcode `\%12\relax}%
\providecommand \@@startlink[1]{}%
\providecommand \@@endlink[0]{}%
\providecommand \url  [0]{\begingroup\@sanitize@url \@url }%
\providecommand \@url [1]{\endgroup\@href {#1}{\urlprefix }}%
\providecommand \urlprefix  [0]{URL }%
\providecommand \Eprint [0]{\href }%
\providecommand \doibase [0]{https://doi.org/}%
\providecommand \selectlanguage [0]{\@gobble}%
\providecommand \bibinfo  [0]{\@secondoftwo}%
\providecommand \bibfield  [0]{\@secondoftwo}%
\providecommand \translation [1]{[#1]}%
\providecommand \BibitemOpen [0]{}%
\providecommand \bibitemStop [0]{}%
\providecommand \bibitemNoStop [0]{.\EOS\space}%
\providecommand \EOS [0]{\spacefactor3000\relax}%
\providecommand \BibitemShut  [1]{\csname bibitem#1\endcsname}%
\let\auto@bib@innerbib\@empty
\bibitem [{\citenamefont {Shor}(1996)}]{Shor-FT-1996}%
  \BibitemOpen
  \bibfield  {author} {\bibinfo {author} {\bibfnamefont {P.~W.}\ \bibnamefont
  {Shor}},\ }\bibfield  {title} {\bibinfo {title} {Fault-tolerant quantum
  computation},\ }in\ \href {http://arxiv.org/abs/quant-ph/9605011v2} {\emph
  {\bibinfo {booktitle} {Proc. 37th Ann. Symp. on Fundamentals of Comp.
  Sci.}}},\ \bibinfo {organization} {IEEE}\ (\bibinfo  {publisher} {IEEE Comp.
  Soc. Press},\ \bibinfo {address} {Los Alamitos},\ \bibinfo {year} {1996})\
  pp.\ \bibinfo {pages} {56--65},\ \Eprint
  {https://arxiv.org/abs/quant-ph/9605011} {quant-ph/9605011} \BibitemShut
  {NoStop}%
\bibitem [{\citenamefont {Gottesman}(1997)}]{gottesman-thesis}%
  \BibitemOpen
  \bibfield  {author} {\bibinfo {author} {\bibfnamefont {D.}~\bibnamefont
  {Gottesman}},\ }\emph {\bibinfo {title} {Stabilizer Codes and Quantum Error
  Correction}},\ \href {http://arxiv.org/abs/quant-ph/9705052} {Ph.D. thesis},\
  \bibinfo  {school} {Caltech} (\bibinfo {year} {1997})\BibitemShut {NoStop}%
\bibitem [{\citenamefont {Kitaev}(2003)}]{kitaev-anyons}%
  \BibitemOpen
  \bibfield  {author} {\bibinfo {author} {\bibfnamefont {A.~Y.}\ \bibnamefont
  {Kitaev}},\ }\bibfield  {title} {\bibinfo {title} {Fault-tolerant quantum
  computation by anyons},\ }\href {http://arxiv.org/abs/quant-ph/9707021}
  {\bibfield  {journal} {\bibinfo  {journal} {Ann. Phys.}\ }\textbf {\bibinfo
  {volume} {303}},\ \bibinfo {pages} {2} (\bibinfo {year} {2003})}\BibitemShut
  {NoStop}%
\bibitem [{\citenamefont {Bravyi}\ and\ \citenamefont
  {Kitaev}(1998)}]{Bravyi-Kitaev-1998}%
  \BibitemOpen
  \bibfield  {author} {\bibinfo {author} {\bibfnamefont {S.~B.}\ \bibnamefont
  {Bravyi}}\ and\ \bibinfo {author} {\bibfnamefont {A.~Y.}\ \bibnamefont
  {Kitaev}},\ }\bibfield  {title} {\bibinfo {title} {Quantum codes on a lattice
  with boundary},\ }\Eprint {https://arxiv.org/abs/quant-ph/9811052}
  {quant-ph/9811052}  (\bibinfo {year} {1998}),\ \bibinfo {note}
  {unpublished}\BibitemShut {NoStop}%
\bibitem [{\citenamefont {Dennis}\ \emph {et~al.}(2002)\citenamefont {Dennis},
  \citenamefont {Kitaev}, \citenamefont {Landahl},\ and\ \citenamefont
  {Preskill}}]{Dennis-Kitaev-Landahl-Preskill-2002}%
  \BibitemOpen
  \bibfield  {author} {\bibinfo {author} {\bibfnamefont {E.}~\bibnamefont
  {Dennis}}, \bibinfo {author} {\bibfnamefont {A.}~\bibnamefont {Kitaev}},
  \bibinfo {author} {\bibfnamefont {A.}~\bibnamefont {Landahl}},\ and\ \bibinfo
  {author} {\bibfnamefont {J.}~\bibnamefont {Preskill}},\ }\bibfield  {title}
  {\bibinfo {title} {Topological quantum memory},\ }\href
  {https://doi.org/10.1063/1.1499754} {\bibfield  {journal} {\bibinfo
  {journal} {J. Math. Phys.}\ }\textbf {\bibinfo {volume} {43}},\ \bibinfo
  {pages} {4452} (\bibinfo {year} {2002})}\BibitemShut {NoStop}%
\bibitem [{\citenamefont {{Google Quantum
  AI}}(2023)}]{google-2023-suppressing}%
  \BibitemOpen
  \bibfield  {author} {\bibinfo {author} {\bibnamefont {{Google Quantum AI}}},\
  }\bibfield  {title} {\bibinfo {title} {Suppressing quantum errors by scaling
  a surface code logical qubit},\ }\href
  {https://doi.org/10.1038/s41586-022-05434-1} {\bibfield  {journal} {\bibinfo
  {journal} {Nature}\ }\textbf {\bibinfo {volume} {614}},\ \bibinfo {pages}
  {676} (\bibinfo {year} {2023})},\ \Eprint {https://arxiv.org/abs/2207.06431}
  {arXiv:2207.06431 [quant-ph]} \BibitemShut {NoStop}%
\bibitem [{\citenamefont {Paetznick}\ \emph {et~al.}(2024)\citenamefont
  {Paetznick}, \citenamefont {da~Silva}, \citenamefont {Ryan-Anderson},
  \citenamefont {Bello-Rivas}, \citenamefont {Campora~III}, \citenamefont
  {Chernoguzov}, \citenamefont {Dreiling}, \citenamefont {Foltz}, \citenamefont
  {Frachon}, \citenamefont {Gaebler}, \citenamefont {Gatterman}, \citenamefont
  {Grans-Samuelsson}, \citenamefont {Gresh}, \citenamefont {Hayes},
  \citenamefont {Hewitt}, \citenamefont {Holliman}, \citenamefont {Horst},
  \citenamefont {Johansen}, \citenamefont {Lucchetti}, \citenamefont
  {Matsuoka}, \citenamefont {Mills}, \citenamefont {Moses}, \citenamefont
  {Neyenhuis}, \citenamefont {Paz}, \citenamefont {Pino}, \citenamefont
  {Siegfried}, \citenamefont {Sundaram}, \citenamefont {Tom}, \citenamefont
  {Wernli}, \citenamefont {Zanner}, \citenamefont {Stutz},\ and\ \citenamefont
  {Svore}}]{Paetznik-etal-trapped-ions-2024}%
  \BibitemOpen
  \bibfield  {author} {\bibinfo {author} {\bibfnamefont {A.}~\bibnamefont
  {Paetznick}}, \bibinfo {author} {\bibfnamefont {M.~P.}\ \bibnamefont
  {da~Silva}}, \bibinfo {author} {\bibfnamefont {C.}~\bibnamefont
  {Ryan-Anderson}}, \bibinfo {author} {\bibfnamefont {J.~M.}\ \bibnamefont
  {Bello-Rivas}}, \bibinfo {author} {\bibfnamefont {J.~P.}\ \bibnamefont
  {Campora~III}}, \bibinfo {author} {\bibfnamefont {A.}~\bibnamefont
  {Chernoguzov}}, \bibinfo {author} {\bibfnamefont {J.~M.}\ \bibnamefont
  {Dreiling}}, \bibinfo {author} {\bibfnamefont {C.}~\bibnamefont {Foltz}},
  \bibinfo {author} {\bibfnamefont {F.}~\bibnamefont {Frachon}}, \bibinfo
  {author} {\bibfnamefont {J.~P.}\ \bibnamefont {Gaebler}}, \bibinfo {author}
  {\bibfnamefont {T.~M.}\ \bibnamefont {Gatterman}}, \bibinfo {author}
  {\bibfnamefont {L.}~\bibnamefont {Grans-Samuelsson}}, \bibinfo {author}
  {\bibfnamefont {D.}~\bibnamefont {Gresh}}, \bibinfo {author} {\bibfnamefont
  {D.}~\bibnamefont {Hayes}}, \bibinfo {author} {\bibfnamefont
  {N.}~\bibnamefont {Hewitt}}, \bibinfo {author} {\bibfnamefont
  {C.}~\bibnamefont {Holliman}}, \bibinfo {author} {\bibfnamefont {C.~V.}\
  \bibnamefont {Horst}}, \bibinfo {author} {\bibfnamefont {J.}~\bibnamefont
  {Johansen}}, \bibinfo {author} {\bibfnamefont {D.}~\bibnamefont {Lucchetti}},
  \bibinfo {author} {\bibfnamefont {Y.}~\bibnamefont {Matsuoka}}, \bibinfo
  {author} {\bibfnamefont {M.}~\bibnamefont {Mills}}, \bibinfo {author}
  {\bibfnamefont {S.~A.}\ \bibnamefont {Moses}}, \bibinfo {author}
  {\bibfnamefont {B.}~\bibnamefont {Neyenhuis}}, \bibinfo {author}
  {\bibfnamefont {A.}~\bibnamefont {Paz}}, \bibinfo {author} {\bibfnamefont
  {J.}~\bibnamefont {Pino}}, \bibinfo {author} {\bibfnamefont {P.}~\bibnamefont
  {Siegfried}}, \bibinfo {author} {\bibfnamefont {A.}~\bibnamefont {Sundaram}},
  \bibinfo {author} {\bibfnamefont {D.}~\bibnamefont {Tom}}, \bibinfo {author}
  {\bibfnamefont {S.~J.}\ \bibnamefont {Wernli}}, \bibinfo {author}
  {\bibfnamefont {M.}~\bibnamefont {Zanner}}, \bibinfo {author} {\bibfnamefont
  {R.~P.}\ \bibnamefont {Stutz}},\ and\ \bibinfo {author} {\bibfnamefont
  {K.~M.}\ \bibnamefont {Svore}},\ }\bibfield  {title} {\bibinfo {title}
  {Demonstration of logical qubits and repeated error correction with
  better-than-physical error rates},\ }\Eprint
  {https://arxiv.org/abs/2404.02280} {2404.02280}  (\bibinfo {year} {2024}),\
  \bibinfo {note} {unpublished}\BibitemShut {NoStop}%
\bibitem [{\citenamefont {Bluvstein}\ \emph {et~al.}(2024)\citenamefont
  {Bluvstein}, \citenamefont {Evered}, \citenamefont {Geim}, \citenamefont
  {Li}, \citenamefont {Zhou}, \citenamefont {Manovitz}, \citenamefont {Ebadi},
  \citenamefont {Cain}, \citenamefont {Kalinowski}, \citenamefont {Hangleiter},
  \citenamefont {Ataides}, \citenamefont {Maskara}, \citenamefont {Cong},
  \citenamefont {Gao}, \citenamefont {Rodriguez}, \citenamefont {Karolyshyn},
  \citenamefont {Semeghini}, \citenamefont {Gullans}, \citenamefont {Greiner},
  \citenamefont {Vuleti\'c},\ and\ \citenamefont
  {Lukin}}]{Bluvstein-etal-Lukin-2024}%
  \BibitemOpen
  \bibfield  {author} {\bibinfo {author} {\bibfnamefont {D.}~\bibnamefont
  {Bluvstein}}, \bibinfo {author} {\bibfnamefont {S.~J.}\ \bibnamefont
  {Evered}}, \bibinfo {author} {\bibfnamefont {A.~A.}\ \bibnamefont {Geim}},
  \bibinfo {author} {\bibfnamefont {S.~H.}\ \bibnamefont {Li}}, \bibinfo
  {author} {\bibfnamefont {H.}~\bibnamefont {Zhou}}, \bibinfo {author}
  {\bibfnamefont {T.}~\bibnamefont {Manovitz}}, \bibinfo {author}
  {\bibfnamefont {S.}~\bibnamefont {Ebadi}}, \bibinfo {author} {\bibfnamefont
  {M.}~\bibnamefont {Cain}}, \bibinfo {author} {\bibfnamefont {M.}~\bibnamefont
  {Kalinowski}}, \bibinfo {author} {\bibfnamefont {D.}~\bibnamefont
  {Hangleiter}}, \bibinfo {author} {\bibfnamefont {J.~P.~B.}\ \bibnamefont
  {Ataides}}, \bibinfo {author} {\bibfnamefont {N.}~\bibnamefont {Maskara}},
  \bibinfo {author} {\bibfnamefont {I.}~\bibnamefont {Cong}}, \bibinfo {author}
  {\bibfnamefont {X.}~\bibnamefont {Gao}}, \bibinfo {author} {\bibfnamefont
  {P.~S.}\ \bibnamefont {Rodriguez}}, \bibinfo {author} {\bibfnamefont
  {T.}~\bibnamefont {Karolyshyn}}, \bibinfo {author} {\bibfnamefont
  {G.}~\bibnamefont {Semeghini}}, \bibinfo {author} {\bibfnamefont {M.~J.}\
  \bibnamefont {Gullans}}, \bibinfo {author} {\bibfnamefont {M.}~\bibnamefont
  {Greiner}}, \bibinfo {author} {\bibfnamefont {V.}~\bibnamefont {Vuleti\'c}},\
  and\ \bibinfo {author} {\bibfnamefont {M.~D.}\ \bibnamefont {Lukin}},\
  }\bibfield  {title} {\bibinfo {title} {Logical quantum processor based on
  reconfigurable atom arrays},\ }\href
  {https://doi.org/https://doi.org/10.1038/s41586-023-06927-3} {\bibfield
  {journal} {\bibinfo  {journal} {Nature}\ }\textbf {\bibinfo {volume} {626}},\
  \bibinfo {pages} {58} (\bibinfo {year} {2024})}\BibitemShut {NoStop}%
\bibitem [{\citenamefont {Acharya}\ \emph {et~al.}(2025)\citenamefont
  {Acharya}, \citenamefont {Abanin}, \citenamefont {Aghababaie-Beni},
  \citenamefont {Aleiner}, \citenamefont {Andersen}, \citenamefont {Ansmann},
  \citenamefont {Arute}, \citenamefont {Arya}, \citenamefont {Asfaw},
  \citenamefont {Astrakhantsev}, \citenamefont {Atalaya}, \citenamefont
  {Babbush}, \citenamefont {Bacon}, \citenamefont {Ballard}, \citenamefont
  {Bardin}, \citenamefont {Bausch}, \citenamefont {Bengtsson}, \citenamefont
  {Bilmes}, \citenamefont {Blackwell}, \citenamefont {Boixo}, \citenamefont
  {Bortoli}, \citenamefont {Bourassa}, \citenamefont {Bovaird}, \citenamefont
  {Brill}, \citenamefont {Broughton}, \citenamefont {Browne}, \citenamefont
  {Buchea}, \citenamefont {Buckley}, \citenamefont {Buell}, \citenamefont
  {Burger}, \citenamefont {Burkett}, \citenamefont {Bushnell}, \citenamefont
  {Cabrera}, \citenamefont {Campero}, \citenamefont {Chang}, \citenamefont
  {Chen}, \citenamefont {Chen}, \citenamefont {Chiaro}, \citenamefont {Chik},
  \citenamefont {Chou}, \citenamefont {Claes}, \citenamefont {Cleland},
  \citenamefont {Cogan}, \citenamefont {Collins}, \citenamefont {Conner},
  \citenamefont {Courtney}, \citenamefont {Crook}, \citenamefont {Curtin},
  \citenamefont {Das}, \citenamefont {Davies}, \citenamefont {Lorenzo},
  \citenamefont {Debroy}, \citenamefont {Demura}, \citenamefont {Devoret},
  \citenamefont {Paolo}, \citenamefont {Donohoe}, \citenamefont {Drozdov},
  \citenamefont {Dunsworth}, \citenamefont {Earle}, \citenamefont {Edlich},
  \citenamefont {Eickbusch}, \citenamefont {Elbag}, \citenamefont {Elzouka},
  \citenamefont {Erickson}, \citenamefont {Faoro}, \citenamefont {Farhi},
  \citenamefont {Ferreira}, \citenamefont {Burgos}, \citenamefont {Forati},
  \citenamefont {Fowler}, \citenamefont {Foxen}, \citenamefont {Ganjam},
  \citenamefont {Garcia}, \citenamefont {Gasca}, \citenamefont {Élie Genois},
  \citenamefont {Giang}, \citenamefont {Gidney}, \citenamefont {Gilboa},
  \citenamefont {Gosula}, \citenamefont {Dau}, \citenamefont {Graumann},
  \citenamefont {Greene}, \citenamefont {Gross}, \citenamefont {Habegger},
  \citenamefont {Hall}, \citenamefont {Hamilton}, \citenamefont {Hansen},
  \citenamefont {Harrigan}, \citenamefont {Harrington}, \citenamefont {Heras},
  \citenamefont {Heslin}, \citenamefont {Heu}, \citenamefont {Higgott},
  \citenamefont {Hill}, \citenamefont {Hilton}, \citenamefont {Holland},
  \citenamefont {Hong}, \citenamefont {Huang}, \citenamefont {Huff},
  \citenamefont {Huggins}, \citenamefont {Ioffe}, \citenamefont {Isakov},
  \citenamefont {Iveland}, \citenamefont {Jeffrey}, \citenamefont {Jiang},
  \citenamefont {Jones}, \citenamefont {Jordan}, \citenamefont {Joshi},
  \citenamefont {Juhas}, \citenamefont {Kafri}, \citenamefont {Kang},
  \citenamefont {Karamlou}, \citenamefont {Kechedzhi}, \citenamefont {Kelly},
  \citenamefont {Khaire}, \citenamefont {Khattar}, \citenamefont {Khezri},
  \citenamefont {Kim}, \citenamefont {Klimov}, \citenamefont {Klots},
  \citenamefont {Kobrin}, \citenamefont {Kohli}, \citenamefont {Korotkov},
  \citenamefont {Kostritsa}, \citenamefont {Kothari}, \citenamefont
  {Kozlovskii}, \citenamefont {Kreikebaum}, \citenamefont {Kurilovich},
  \citenamefont {Lacroix}, \citenamefont {Landhuis}, \citenamefont {Lange-Dei},
  \citenamefont {Langley}, \citenamefont {Laptev}, \citenamefont {Lau},
  \citenamefont {Guevel}, \citenamefont {Ledford}, \citenamefont {Lee},
  \citenamefont {Lee}, \citenamefont {Lensky}, \citenamefont {Leon},
  \citenamefont {Lester}, \citenamefont {Li}, \citenamefont {Li}, \citenamefont
  {Lill}, \citenamefont {Liu}, \citenamefont {Livingston}, \citenamefont
  {Locharla}, \citenamefont {Lucero}, \citenamefont {Lundahl},\ and\
  \citenamefont {Lunt}}]{google_Quantum_AI-Y1-2025}%
  \BibitemOpen
  \bibfield  {author} {\bibinfo {author} {\bibfnamefont {R.}~\bibnamefont
  {Acharya}}, \bibinfo {author} {\bibfnamefont {D.~A.}\ \bibnamefont {Abanin}},
  \bibinfo {author} {\bibfnamefont {L.}~\bibnamefont {Aghababaie-Beni}},
  \bibinfo {author} {\bibfnamefont {I.}~\bibnamefont {Aleiner}}, \bibinfo
  {author} {\bibfnamefont {T.~I.}\ \bibnamefont {Andersen}}, \bibinfo {author}
  {\bibfnamefont {M.}~\bibnamefont {Ansmann}}, \bibinfo {author} {\bibfnamefont
  {F.}~\bibnamefont {Arute}}, \bibinfo {author} {\bibfnamefont
  {K.}~\bibnamefont {Arya}}, \bibinfo {author} {\bibfnamefont {A.}~\bibnamefont
  {Asfaw}}, \bibinfo {author} {\bibfnamefont {N.}~\bibnamefont
  {Astrakhantsev}}, \bibinfo {author} {\bibfnamefont {J.}~\bibnamefont
  {Atalaya}}, \bibinfo {author} {\bibfnamefont {R.}~\bibnamefont {Babbush}},
  \bibinfo {author} {\bibfnamefont {D.}~\bibnamefont {Bacon}}, \bibinfo
  {author} {\bibfnamefont {B.}~\bibnamefont {Ballard}}, \bibinfo {author}
  {\bibfnamefont {J.~C.}\ \bibnamefont {Bardin}}, \bibinfo {author}
  {\bibfnamefont {J.}~\bibnamefont {Bausch}}, \bibinfo {author} {\bibfnamefont
  {A.}~\bibnamefont {Bengtsson}}, \bibinfo {author} {\bibfnamefont
  {A.}~\bibnamefont {Bilmes}}, \bibinfo {author} {\bibfnamefont
  {S.}~\bibnamefont {Blackwell}}, \bibinfo {author} {\bibfnamefont
  {S.}~\bibnamefont {Boixo}}, \bibinfo {author} {\bibfnamefont
  {G.}~\bibnamefont {Bortoli}}, \bibinfo {author} {\bibfnamefont
  {A.}~\bibnamefont {Bourassa}}, \bibinfo {author} {\bibfnamefont
  {J.}~\bibnamefont {Bovaird}}, \bibinfo {author} {\bibfnamefont
  {L.}~\bibnamefont {Brill}}, \bibinfo {author} {\bibfnamefont
  {M.}~\bibnamefont {Broughton}}, \bibinfo {author} {\bibfnamefont {D.~A.}\
  \bibnamefont {Browne}}, \bibinfo {author} {\bibfnamefont {B.}~\bibnamefont
  {Buchea}}, \bibinfo {author} {\bibfnamefont {B.~B.}\ \bibnamefont {Buckley}},
  \bibinfo {author} {\bibfnamefont {D.~A.}\ \bibnamefont {Buell}}, \bibinfo
  {author} {\bibfnamefont {T.}~\bibnamefont {Burger}}, \bibinfo {author}
  {\bibfnamefont {B.}~\bibnamefont {Burkett}}, \bibinfo {author} {\bibfnamefont
  {N.}~\bibnamefont {Bushnell}}, \bibinfo {author} {\bibfnamefont
  {A.}~\bibnamefont {Cabrera}}, \bibinfo {author} {\bibfnamefont
  {J.}~\bibnamefont {Campero}}, \bibinfo {author} {\bibfnamefont {H.-S.}\
  \bibnamefont {Chang}}, \bibinfo {author} {\bibfnamefont {Y.}~\bibnamefont
  {Chen}}, \bibinfo {author} {\bibfnamefont {Z.}~\bibnamefont {Chen}}, \bibinfo
  {author} {\bibfnamefont {B.}~\bibnamefont {Chiaro}}, \bibinfo {author}
  {\bibfnamefont {D.}~\bibnamefont {Chik}}, \bibinfo {author} {\bibfnamefont
  {C.}~\bibnamefont {Chou}}, \bibinfo {author} {\bibfnamefont {J.}~\bibnamefont
  {Claes}}, \bibinfo {author} {\bibfnamefont {A.~Y.}\ \bibnamefont {Cleland}},
  \bibinfo {author} {\bibfnamefont {J.}~\bibnamefont {Cogan}}, \bibinfo
  {author} {\bibfnamefont {R.}~\bibnamefont {Collins}}, \bibinfo {author}
  {\bibfnamefont {P.}~\bibnamefont {Conner}}, \bibinfo {author} {\bibfnamefont
  {W.}~\bibnamefont {Courtney}}, \bibinfo {author} {\bibfnamefont {A.~L.}\
  \bibnamefont {Crook}}, \bibinfo {author} {\bibfnamefont {B.}~\bibnamefont
  {Curtin}}, \bibinfo {author} {\bibfnamefont {S.}~\bibnamefont {Das}},
  \bibinfo {author} {\bibfnamefont {A.}~\bibnamefont {Davies}}, \bibinfo
  {author} {\bibfnamefont {L.~D.}\ \bibnamefont {Lorenzo}}, \bibinfo {author}
  {\bibfnamefont {D.~M.}\ \bibnamefont {Debroy}}, \bibinfo {author}
  {\bibfnamefont {S.}~\bibnamefont {Demura}}, \bibinfo {author} {\bibfnamefont
  {M.}~\bibnamefont {Devoret}}, \bibinfo {author} {\bibfnamefont {A.~D.}\
  \bibnamefont {Paolo}}, \bibinfo {author} {\bibfnamefont {P.}~\bibnamefont
  {Donohoe}}, \bibinfo {author} {\bibfnamefont {I.}~\bibnamefont {Drozdov}},
  \bibinfo {author} {\bibfnamefont {A.}~\bibnamefont {Dunsworth}}, \bibinfo
  {author} {\bibfnamefont {C.}~\bibnamefont {Earle}}, \bibinfo {author}
  {\bibfnamefont {T.}~\bibnamefont {Edlich}}, \bibinfo {author} {\bibfnamefont
  {A.}~\bibnamefont {Eickbusch}}, \bibinfo {author} {\bibfnamefont {A.~M.}\
  \bibnamefont {Elbag}}, \bibinfo {author} {\bibfnamefont {M.}~\bibnamefont
  {Elzouka}}, \bibinfo {author} {\bibfnamefont {C.}~\bibnamefont {Erickson}},
  \bibinfo {author} {\bibfnamefont {L.}~\bibnamefont {Faoro}}, \bibinfo
  {author} {\bibfnamefont {E.}~\bibnamefont {Farhi}}, \bibinfo {author}
  {\bibfnamefont {V.~S.}\ \bibnamefont {Ferreira}}, \bibinfo {author}
  {\bibfnamefont {L.~F.}\ \bibnamefont {Burgos}}, \bibinfo {author}
  {\bibfnamefont {E.}~\bibnamefont {Forati}}, \bibinfo {author} {\bibfnamefont
  {A.~G.}\ \bibnamefont {Fowler}}, \bibinfo {author} {\bibfnamefont
  {B.}~\bibnamefont {Foxen}}, \bibinfo {author} {\bibfnamefont
  {S.}~\bibnamefont {Ganjam}}, \bibinfo {author} {\bibfnamefont
  {G.}~\bibnamefont {Garcia}}, \bibinfo {author} {\bibfnamefont
  {R.}~\bibnamefont {Gasca}}, \bibinfo {author} {\bibnamefont {Élie Genois}},
  \bibinfo {author} {\bibfnamefont {W.}~\bibnamefont {Giang}}, \bibinfo
  {author} {\bibfnamefont {C.}~\bibnamefont {Gidney}}, \bibinfo {author}
  {\bibfnamefont {D.}~\bibnamefont {Gilboa}}, \bibinfo {author} {\bibfnamefont
  {R.}~\bibnamefont {Gosula}}, \bibinfo {author} {\bibfnamefont {A.~G.}\
  \bibnamefont {Dau}}, \bibinfo {author} {\bibfnamefont {D.}~\bibnamefont
  {Graumann}}, \bibinfo {author} {\bibfnamefont {A.}~\bibnamefont {Greene}},
  \bibinfo {author} {\bibfnamefont {J.~A.}\ \bibnamefont {Gross}}, \bibinfo
  {author} {\bibfnamefont {S.}~\bibnamefont {Habegger}}, \bibinfo {author}
  {\bibfnamefont {J.}~\bibnamefont {Hall}}, \bibinfo {author} {\bibfnamefont
  {M.~C.}\ \bibnamefont {Hamilton}}, \bibinfo {author} {\bibfnamefont
  {M.}~\bibnamefont {Hansen}}, \bibinfo {author} {\bibfnamefont {M.~P.}\
  \bibnamefont {Harrigan}}, \bibinfo {author} {\bibfnamefont {S.~D.}\
  \bibnamefont {Harrington}}, \bibinfo {author} {\bibfnamefont {F.~J.}\
  \bibnamefont {Heras}}, \bibinfo {author} {\bibfnamefont {S.}~\bibnamefont
  {Heslin}}, \bibinfo {author} {\bibfnamefont {P.}~\bibnamefont {Heu}},
  \bibinfo {author} {\bibfnamefont {O.}~\bibnamefont {Higgott}}, \bibinfo
  {author} {\bibfnamefont {G.}~\bibnamefont {Hill}}, \bibinfo {author}
  {\bibfnamefont {J.}~\bibnamefont {Hilton}}, \bibinfo {author} {\bibfnamefont
  {G.}~\bibnamefont {Holland}}, \bibinfo {author} {\bibfnamefont
  {S.}~\bibnamefont {Hong}}, \bibinfo {author} {\bibfnamefont {H.-Y.}\
  \bibnamefont {Huang}}, \bibinfo {author} {\bibfnamefont {A.}~\bibnamefont
  {Huff}}, \bibinfo {author} {\bibfnamefont {W.~J.}\ \bibnamefont {Huggins}},
  \bibinfo {author} {\bibfnamefont {L.~B.}\ \bibnamefont {Ioffe}}, \bibinfo
  {author} {\bibfnamefont {S.~V.}\ \bibnamefont {Isakov}}, \bibinfo {author}
  {\bibfnamefont {J.}~\bibnamefont {Iveland}}, \bibinfo {author} {\bibfnamefont
  {E.}~\bibnamefont {Jeffrey}}, \bibinfo {author} {\bibfnamefont
  {Z.}~\bibnamefont {Jiang}}, \bibinfo {author} {\bibfnamefont
  {C.}~\bibnamefont {Jones}}, \bibinfo {author} {\bibfnamefont
  {S.}~\bibnamefont {Jordan}}, \bibinfo {author} {\bibfnamefont
  {C.}~\bibnamefont {Joshi}}, \bibinfo {author} {\bibfnamefont
  {P.}~\bibnamefont {Juhas}}, \bibinfo {author} {\bibfnamefont
  {D.}~\bibnamefont {Kafri}}, \bibinfo {author} {\bibfnamefont
  {H.}~\bibnamefont {Kang}}, \bibinfo {author} {\bibfnamefont {A.~H.}\
  \bibnamefont {Karamlou}}, \bibinfo {author} {\bibfnamefont {K.}~\bibnamefont
  {Kechedzhi}}, \bibinfo {author} {\bibfnamefont {J.}~\bibnamefont {Kelly}},
  \bibinfo {author} {\bibfnamefont {T.}~\bibnamefont {Khaire}}, \bibinfo
  {author} {\bibfnamefont {T.}~\bibnamefont {Khattar}}, \bibinfo {author}
  {\bibfnamefont {M.}~\bibnamefont {Khezri}}, \bibinfo {author} {\bibfnamefont
  {S.}~\bibnamefont {Kim}}, \bibinfo {author} {\bibfnamefont {P.~V.}\
  \bibnamefont {Klimov}}, \bibinfo {author} {\bibfnamefont {A.~R.}\
  \bibnamefont {Klots}}, \bibinfo {author} {\bibfnamefont {B.}~\bibnamefont
  {Kobrin}}, \bibinfo {author} {\bibfnamefont {P.}~\bibnamefont {Kohli}},
  \bibinfo {author} {\bibfnamefont {A.~N.}\ \bibnamefont {Korotkov}}, \bibinfo
  {author} {\bibfnamefont {F.}~\bibnamefont {Kostritsa}}, \bibinfo {author}
  {\bibfnamefont {R.}~\bibnamefont {Kothari}}, \bibinfo {author} {\bibfnamefont
  {B.}~\bibnamefont {Kozlovskii}}, \bibinfo {author} {\bibfnamefont {J.~M.}\
  \bibnamefont {Kreikebaum}}, \bibinfo {author} {\bibfnamefont {V.~D.}\
  \bibnamefont {Kurilovich}}, \bibinfo {author} {\bibfnamefont
  {N.}~\bibnamefont {Lacroix}}, \bibinfo {author} {\bibfnamefont
  {D.}~\bibnamefont {Landhuis}}, \bibinfo {author} {\bibfnamefont
  {T.}~\bibnamefont {Lange-Dei}}, \bibinfo {author} {\bibfnamefont {B.~W.}\
  \bibnamefont {Langley}}, \bibinfo {author} {\bibfnamefont {P.}~\bibnamefont
  {Laptev}}, \bibinfo {author} {\bibfnamefont {K.-M.}\ \bibnamefont {Lau}},
  \bibinfo {author} {\bibfnamefont {L.~L.}\ \bibnamefont {Guevel}}, \bibinfo
  {author} {\bibfnamefont {J.}~\bibnamefont {Ledford}}, \bibinfo {author}
  {\bibfnamefont {J.}~\bibnamefont {Lee}}, \bibinfo {author} {\bibfnamefont
  {K.}~\bibnamefont {Lee}}, \bibinfo {author} {\bibfnamefont {Y.~D.}\
  \bibnamefont {Lensky}}, \bibinfo {author} {\bibfnamefont {S.}~\bibnamefont
  {Leon}}, \bibinfo {author} {\bibfnamefont {B.~J.}\ \bibnamefont {Lester}},
  \bibinfo {author} {\bibfnamefont {W.~Y.}\ \bibnamefont {Li}}, \bibinfo
  {author} {\bibfnamefont {Y.}~\bibnamefont {Li}}, \bibinfo {author}
  {\bibfnamefont {A.~T.}\ \bibnamefont {Lill}}, \bibinfo {author}
  {\bibfnamefont {W.}~\bibnamefont {Liu}}, \bibinfo {author} {\bibfnamefont
  {W.~P.}\ \bibnamefont {Livingston}}, \bibinfo {author} {\bibfnamefont
  {A.}~\bibnamefont {Locharla}}, \bibinfo {author} {\bibfnamefont
  {E.}~\bibnamefont {Lucero}}, \bibinfo {author} {\bibfnamefont
  {D.}~\bibnamefont {Lundahl}},\ and\ \bibinfo {author} {\bibfnamefont
  {A.}~\bibnamefont {Lunt}},\ }\bibfield  {title} {\bibinfo {title} {Quantum
  error correction below the surface code threshold},\ }\href
  {https://doi.org/10.1038/s41586-024-08449-y} {\bibfield  {journal} {\bibinfo
  {journal} {Nature}\ }\textbf {\bibinfo {volume} {638}},\ \bibinfo {pages}
  {920} (\bibinfo {year} {2025})},\ \Eprint {https://arxiv.org/abs/2408.13687}
  {2408.13687} \BibitemShut {NoStop}%
\bibitem [{\citenamefont {Gong}\ \emph {et~al.}(2024)\citenamefont {Gong},
  \citenamefont {Cammerer},\ and\ \citenamefont
  {Renes}}]{Gong-Cammerer-Renes-2024}%
  \BibitemOpen
  \bibfield  {author} {\bibinfo {author} {\bibfnamefont {A.}~\bibnamefont
  {Gong}}, \bibinfo {author} {\bibfnamefont {S.}~\bibnamefont {Cammerer}},\
  and\ \bibinfo {author} {\bibfnamefont {J.~M.}\ \bibnamefont {Renes}},\
  }\bibfield  {title} {\bibinfo {title} {Toward low-latency iterative decoding
  of qldpc codes under circuit-level noise},\ }\Eprint
  {https://arxiv.org/abs/2403.18901} {arXiv:2403.18901 [quant-ph]}  (\bibinfo
  {year} {2024}),\ \bibinfo {note} {unpublished}\BibitemShut {NoStop}%
\bibitem [{\citenamefont {Skoric}\ \emph {et~al.}(2023)\citenamefont {Skoric},
  \citenamefont {Browne}, \citenamefont {Barnes}, \citenamefont {Gillespie},\
  and\ \citenamefont
  {Campbell}}]{Skoric-Browne-Barnes-Gillespie-Campbell-2023}%
  \BibitemOpen
  \bibfield  {author} {\bibinfo {author} {\bibfnamefont {L.}~\bibnamefont
  {Skoric}}, \bibinfo {author} {\bibfnamefont {D.~E.}\ \bibnamefont {Browne}},
  \bibinfo {author} {\bibfnamefont {K.~M.}\ \bibnamefont {Barnes}}, \bibinfo
  {author} {\bibfnamefont {N.~I.}\ \bibnamefont {Gillespie}},\ and\ \bibinfo
  {author} {\bibfnamefont {E.~T.}\ \bibnamefont {Campbell}},\ }\bibfield
  {title} {\bibinfo {title} {Parallel window decoding enables scalable fault
  tolerant quantum computation},\ }\href
  {https://doi.org/10.1038/s41467-023-42482-1} {\bibfield  {journal} {\bibinfo
  {journal} {Nature Communications}\ }\textbf {\bibinfo {volume} {14}},\
  \bibinfo {pages} {7040} (\bibinfo {year} {2023})}\BibitemShut {NoStop}%
\bibitem [{\citenamefont {Bomb\'{\i}n}(2015)}]{Bombin-2015}%
  \BibitemOpen
  \bibfield  {author} {\bibinfo {author} {\bibfnamefont {H.}~\bibnamefont
  {Bomb\'{\i}n}},\ }\bibfield  {title} {\bibinfo {title} {Single-shot
  fault-tolerant quantum error correction},\ }\href
  {https://doi.org/10.1103/PhysRevX.5.031043} {\bibfield  {journal} {\bibinfo
  {journal} {Phys. Rev. X}\ }\textbf {\bibinfo {volume} {5}},\ \bibinfo {pages}
  {031043} (\bibinfo {year} {2015})},\ \Eprint
  {https://arxiv.org/abs/arXiv:1404.5504} {arXiv:1404.5504} \BibitemShut
  {NoStop}%
\bibitem [{\citenamefont {Brown}\ \emph {et~al.}(2016)\citenamefont {Brown},
  \citenamefont {Nickerson},\ and\ \citenamefont
  {Browne}}]{Brown-Nickerson-Browne-2016}%
  \BibitemOpen
  \bibfield  {author} {\bibinfo {author} {\bibfnamefont {B.~J.}\ \bibnamefont
  {Brown}}, \bibinfo {author} {\bibfnamefont {N.~H.}\ \bibnamefont
  {Nickerson}},\ and\ \bibinfo {author} {\bibfnamefont {D.~E.}\ \bibnamefont
  {Browne}},\ }\bibfield  {title} {\bibinfo {title} {Fault-tolerant error
  correction with the gauge color code},\ }\href
  {https://doi.org/10.1038/ncomms12302} {\bibfield  {journal} {\bibinfo
  {journal} {Nature Communications}\ }\textbf {\bibinfo {volume} {7}},\
  \bibinfo {pages} {12302} (\bibinfo {year} {2016})}\BibitemShut {NoStop}%
\bibitem [{\citenamefont {Campbell}(2019)}]{Campbell-2018}%
  \BibitemOpen
  \bibfield  {author} {\bibinfo {author} {\bibfnamefont {E.~T.}\ \bibnamefont
  {Campbell}},\ }\bibfield  {title} {\bibinfo {title} {A theory of single-shot
  error correction for adversarial noise},\ }\href
  {https://doi.org/10.1088/2058-9565/aafc8f} {\bibfield  {journal} {\bibinfo
  {journal} {Quantum Science and Technology}\ }\textbf {\bibinfo {volume}
  {4}},\ \bibinfo {pages} {025006} (\bibinfo {year} {2019})},\ \Eprint
  {https://arxiv.org/abs/1805.09271} {1805.09271} \BibitemShut {NoStop}%
\bibitem [{\citenamefont {Fujiwara}(2014)}]{Fujiwara-2014}%
  \BibitemOpen
  \bibfield  {author} {\bibinfo {author} {\bibfnamefont {Y.}~\bibnamefont
  {Fujiwara}},\ }\bibfield  {title} {\bibinfo {title} {Ability of stabilizer
  quantum error correction to protect itself from its own imperfection},\
  }\href {https://doi.org/10.1103/PhysRevA.90.062304} {\bibfield  {journal}
  {\bibinfo  {journal} {Phys. Rev. A}\ }\textbf {\bibinfo {volume} {90}},\
  \bibinfo {pages} {062304} (\bibinfo {year} {2014})}\BibitemShut {NoStop}%
\bibitem [{\citenamefont {Ashikhmin}\ \emph {et~al.}(2014)\citenamefont
  {Ashikhmin}, \citenamefont {Lai},\ and\ \citenamefont
  {Brun}}]{Ashikhmin-Lai-Brun-2014}%
  \BibitemOpen
  \bibfield  {author} {\bibinfo {author} {\bibfnamefont {A.}~\bibnamefont
  {Ashikhmin}}, \bibinfo {author} {\bibfnamefont {C.~Y.}\ \bibnamefont {Lai}},\
  and\ \bibinfo {author} {\bibfnamefont {T.~A.}\ \bibnamefont {Brun}},\
  }\bibfield  {title} {\bibinfo {title} {Robust quantum error syndrome
  extraction by classical coding},\ }in\ \href
  {https://doi.org/10.1109/ISIT.2014.6874892} {\emph {\bibinfo {booktitle}
  {2014 {IEEE} International Symposium on Information Theory}}}\ (\bibinfo
  {year} {2014})\ pp.\ \bibinfo {pages} {546--550}\BibitemShut {NoStop}%
\bibitem [{\citenamefont {Ashikhmin}\ \emph {et~al.}(2016)\citenamefont
  {Ashikhmin}, \citenamefont {Lai},\ and\ \citenamefont
  {Brun}}]{Ashikhmin-Lai-Brun-2016}%
  \BibitemOpen
  \bibfield  {author} {\bibinfo {author} {\bibfnamefont {A.}~\bibnamefont
  {Ashikhmin}}, \bibinfo {author} {\bibfnamefont {C.~Y.}\ \bibnamefont {Lai}},\
  and\ \bibinfo {author} {\bibfnamefont {T.~A.}\ \bibnamefont {Brun}},\
  }\bibfield  {title} {\bibinfo {title} {Correction of data and syndrome errors
  by stabilizer codes},\ }in\ \href {https://doi.org/10.1109/ISIT.2016.7541704}
  {\emph {\bibinfo {booktitle} {2016 IEEE International Symposium on
  Information Theory (ISIT)}}}\ (\bibinfo {year} {2016})\ pp.\ \bibinfo {pages}
  {2274--2278},\ \Eprint {https://arxiv.org/abs/arXiv:1602.01545}
  {arXiv:1602.01545} \BibitemShut {NoStop}%
\bibitem [{\citenamefont {Breuckmann}\ and\ \citenamefont
  {Londe}(2020)}]{Breuckmann-Londe-2020}%
  \BibitemOpen
  \bibfield  {author} {\bibinfo {author} {\bibfnamefont {N.~P.}\ \bibnamefont
  {Breuckmann}}\ and\ \bibinfo {author} {\bibfnamefont {V.}~\bibnamefont
  {Londe}},\ }\bibfield  {title} {\bibinfo {title} {Single-shot decoding of
  linear rate {LDPC} quantum codes with high performance},\ }\Eprint
  {https://arxiv.org/abs/arXiv:2001.03568} {arXiv:2001.03568}  (\bibinfo {year}
  {2020}),\ \bibinfo {note} {unpublished}\BibitemShut {NoStop}%
\bibitem [{\citenamefont {Breuckmann}\ \emph {et~al.}(2017)\citenamefont
  {Breuckmann}, \citenamefont {Duivenvoorden}, \citenamefont {Michels},\ and\
  \citenamefont {Terhal}}]{Breuckmann-Duivenvoorden-Michels-Terhal-2017}%
  \BibitemOpen
  \bibfield  {author} {\bibinfo {author} {\bibfnamefont {N.~P.}\ \bibnamefont
  {Breuckmann}}, \bibinfo {author} {\bibfnamefont {K.}~\bibnamefont
  {Duivenvoorden}}, \bibinfo {author} {\bibfnamefont {D.}~\bibnamefont
  {Michels}},\ and\ \bibinfo {author} {\bibfnamefont {B.~M.}\ \bibnamefont
  {Terhal}},\ }\bibfield  {title} {\bibinfo {title} {Local decoders for the
  {2D} and {4D} toric code},\ }\href@noop {} {\bibfield  {journal} {\bibinfo
  {journal} {Quantum Inf. Comput.}\ }\textbf {\bibinfo {volume} {17}},\
  \bibinfo {pages} {0181} (\bibinfo {year} {2017})},\ \Eprint
  {https://arxiv.org/abs/1609.00510} {1609.00510} \BibitemShut {NoStop}%
\bibitem [{\citenamefont {Zeng}\ and\ \citenamefont
  {Pryadko}(2019)}]{Zeng-Pryadko-2018}%
  \BibitemOpen
  \bibfield  {author} {\bibinfo {author} {\bibfnamefont {W.}~\bibnamefont
  {Zeng}}\ and\ \bibinfo {author} {\bibfnamefont {L.~P.}\ \bibnamefont
  {Pryadko}},\ }\bibfield  {title} {\bibinfo {title} {Higher-dimensional
  quantum hypergraph-product codes with finite rates},\ }\href
  {https://doi.org/10.1103/PhysRevLett.122.230501} {\bibfield  {journal}
  {\bibinfo  {journal} {Phys. Rev. Lett.}\ }\textbf {\bibinfo {volume} {122}},\
  \bibinfo {pages} {230501} (\bibinfo {year} {2019})},\ \Eprint
  {https://arxiv.org/abs/1810.01519} {1810.01519} \BibitemShut {NoStop}%
\bibitem [{\citenamefont {Zeng}\ and\ \citenamefont
  {Pryadko}(2020)}]{Zeng-Pryadko-hprod-2020}%
  \BibitemOpen
  \bibfield  {author} {\bibinfo {author} {\bibfnamefont {W.}~\bibnamefont
  {Zeng}}\ and\ \bibinfo {author} {\bibfnamefont {L.~P.}\ \bibnamefont
  {Pryadko}},\ }\bibfield  {title} {\bibinfo {title} {Minimal distances for
  certain quantum product codes and tensor products of chain complexes},\
  }\href {https://doi.org/10.1103/PhysRevA.102.062402} {\bibfield  {journal}
  {\bibinfo  {journal} {Phys. Rev. A}\ }\textbf {\bibinfo {volume} {102}},\
  \bibinfo {pages} {062402} (\bibinfo {year} {2020})},\ \Eprint
  {https://arxiv.org/abs/arXiv:2007.12152} {arXiv:2007.12152} \BibitemShut
  {NoStop}%
\bibitem [{\citenamefont {Quintavalle}\ \emph {et~al.}(2021)\citenamefont
  {Quintavalle}, \citenamefont {Vasmer}, \citenamefont {Roffe},\ and\
  \citenamefont {Campbell}}]{Quintavalle-Vasmer-Roffe-Campbell-2021}%
  \BibitemOpen
  \bibfield  {author} {\bibinfo {author} {\bibfnamefont {A.~O.}\ \bibnamefont
  {Quintavalle}}, \bibinfo {author} {\bibfnamefont {M.}~\bibnamefont {Vasmer}},
  \bibinfo {author} {\bibfnamefont {J.}~\bibnamefont {Roffe}},\ and\ \bibinfo
  {author} {\bibfnamefont {E.~T.}\ \bibnamefont {Campbell}},\ }\bibfield
  {title} {\bibinfo {title} {Single-shot error correction of three-dimensional
  homological product codes},\ }\href
  {https://doi.org/10.1103/PRXQuantum.2.020340} {\bibfield  {journal} {\bibinfo
   {journal} {PRX Quantum}\ }\textbf {\bibinfo {volume} {2}},\ \bibinfo {pages}
  {020340} (\bibinfo {year} {2021})}\BibitemShut {NoStop}%
\bibitem [{\citenamefont {Higgott}\ and\ \citenamefont
  {Breuckmann}(2023)}]{Higgott-Breuckmann-2023}%
  \BibitemOpen
  \bibfield  {author} {\bibinfo {author} {\bibfnamefont {O.}~\bibnamefont
  {Higgott}}\ and\ \bibinfo {author} {\bibfnamefont {N.~P.}\ \bibnamefont
  {Breuckmann}},\ }\bibfield  {title} {\bibinfo {title} {Improved single-shot
  decoding of higher-dimensional hypergraph-product codes},\ }\href
  {https://doi.org/10.1103/PRXQuantum.4.020332} {\bibfield  {journal} {\bibinfo
   {journal} {PRX Quantum}\ }\textbf {\bibinfo {volume} {4}},\ \bibinfo {pages}
  {020332} (\bibinfo {year} {2023})}\BibitemShut {NoStop}%
\bibitem [{\citenamefont {Kovalev}\ and\ \citenamefont
  {Pryadko}(2013{\natexlab{a}})}]{Kovalev-Pryadko-Hyperbicycle-2013}%
  \BibitemOpen
  \bibfield  {author} {\bibinfo {author} {\bibfnamefont {A.~A.}\ \bibnamefont
  {Kovalev}}\ and\ \bibinfo {author} {\bibfnamefont {L.~P.}\ \bibnamefont
  {Pryadko}},\ }\bibfield  {title} {\bibinfo {title} {Quantum {K}ronecker
  sum-product low-density parity-check codes with finite rate},\ }\href
  {https://doi.org/10.1103/PhysRevA.88.012311} {\bibfield  {journal} {\bibinfo
  {journal} {Phys. Rev. A}\ }\textbf {\bibinfo {volume} {88}},\ \bibinfo
  {pages} {012311} (\bibinfo {year} {2013}{\natexlab{a}})}\BibitemShut
  {NoStop}%
\bibitem [{\citenamefont {Lin}\ \emph {et~al.}(2025)\citenamefont {Lin},
  \citenamefont {Liu}, \citenamefont {Lim},\ and\ \citenamefont
  {Pryadko}}]{Lin-Liu-Lim-Pryadko-circuits-2024}%
  \BibitemOpen
  \bibfield  {author} {\bibinfo {author} {\bibfnamefont {H.-K.}\ \bibnamefont
  {Lin}}, \bibinfo {author} {\bibfnamefont {X.}~\bibnamefont {Liu}}, \bibinfo
  {author} {\bibfnamefont {P.~K.}\ \bibnamefont {Lim}},\ and\ \bibinfo {author}
  {\bibfnamefont {L.~P.}\ \bibnamefont {Pryadko}},\ }\bibfield  {title}
  {\bibinfo {title} {Single-shot and two-shot decoding with generalized bicycle
  codes},\ }\Eprint {https://arxiv.org/abs/2502.19406} {2502.19406}  (\bibinfo
  {year} {2025}),\ \bibinfo {note} {unpublished}\BibitemShut {NoStop}%
\bibitem [{\citenamefont {Wang}\ and\ \citenamefont
  {Pryadko}(2022)}]{Wang-Pryadko-2022}%
  \BibitemOpen
  \bibfield  {author} {\bibinfo {author} {\bibfnamefont {R.}~\bibnamefont
  {Wang}}\ and\ \bibinfo {author} {\bibfnamefont {L.~P.}\ \bibnamefont
  {Pryadko}},\ }\bibfield  {title} {\bibinfo {title} {Distance bounds for
  generalized bicycle codes},\ }\href {https://doi.org/10.3390/sym14071348}
  {\bibfield  {journal} {\bibinfo  {journal} {Symmetry}\ }\textbf {\bibinfo
  {volume} {14}},\ \bibinfo {pages} {1348} (\bibinfo {year}
  {2022})}\BibitemShut {NoStop}%
\bibitem [{\citenamefont {Wang}\ \emph {et~al.}(2023)\citenamefont {Wang},
  \citenamefont {Lin},\ and\ \citenamefont {Pryadko}}]{Wang-Lin-Pryadko-2023}%
  \BibitemOpen
  \bibfield  {author} {\bibinfo {author} {\bibfnamefont {R.}~\bibnamefont
  {Wang}}, \bibinfo {author} {\bibfnamefont {H.-K.}\ \bibnamefont {Lin}},\ and\
  \bibinfo {author} {\bibfnamefont {L.~P.}\ \bibnamefont {Pryadko}},\
  }\bibfield  {title} {\bibinfo {title} {Abelian and non-abelian quantum
  two-block codes},\ }in\ \href
  {https://doi.org/10.1109/ISTC57237.2023.10273492} {\emph {\bibinfo
  {booktitle} {2023 12th International Symposium on Topics in Coding
  ({ISTC})}}}\ (\bibinfo {year} {2023})\ pp.\ \bibinfo {pages} {1--5},\ \Eprint
  {https://arxiv.org/abs/arXiv:2305.06890} {arXiv:2305.06890} \BibitemShut
  {NoStop}%
\bibitem [{\citenamefont {Lin}\ and\ \citenamefont
  {Pryadko}(2024)}]{Lin-Pryadko-2023}%
  \BibitemOpen
  \bibfield  {author} {\bibinfo {author} {\bibfnamefont {H.-K.}\ \bibnamefont
  {Lin}}\ and\ \bibinfo {author} {\bibfnamefont {L.~P.}\ \bibnamefont
  {Pryadko}},\ }\bibfield  {title} {\bibinfo {title} {Quantum two-block group
  algebra codes},\ }\href {https://doi.org/10.1103/PhysRevA.109.022407}
  {\bibfield  {journal} {\bibinfo  {journal} {Phys. Rev. A}\ }\textbf {\bibinfo
  {volume} {109}},\ \bibinfo {pages} {022407} (\bibinfo {year} {2024})},\
  \Eprint {https://arxiv.org/abs/arXiv:2306.16400} {arXiv:2306.16400}
  \BibitemShut {NoStop}%
\bibitem [{\citenamefont {Bravyi}\ \emph {et~al.}(2023)\citenamefont {Bravyi},
  \citenamefont {Cross}, \citenamefont {Gambetta}, \citenamefont {Maslov},
  \citenamefont {Rall},\ and\ \citenamefont {Yoder}}]{Bravyi-etal-Yoder-2023}%
  \BibitemOpen
  \bibfield  {author} {\bibinfo {author} {\bibfnamefont {S.}~\bibnamefont
  {Bravyi}}, \bibinfo {author} {\bibfnamefont {A.~W.}\ \bibnamefont {Cross}},
  \bibinfo {author} {\bibfnamefont {J.~M.}\ \bibnamefont {Gambetta}}, \bibinfo
  {author} {\bibfnamefont {D.}~\bibnamefont {Maslov}}, \bibinfo {author}
  {\bibfnamefont {P.}~\bibnamefont {Rall}},\ and\ \bibinfo {author}
  {\bibfnamefont {T.~J.}\ \bibnamefont {Yoder}},\ }\bibfield  {title} {\bibinfo
  {title} {High-threshold and low-overhead fault-tolerant quantum memory},\
  }\Eprint {https://arxiv.org/abs/2308.07915} {arXiv:2308.07915 [quant-ph]}
  (\bibinfo {year} {2023}),\ \bibinfo {note} {unpublished}\BibitemShut
  {NoStop}%
\bibitem [{\citenamefont {Liang}\ \emph {et~al.}(2025)\citenamefont {Liang},
  \citenamefont {Liu}, \citenamefont {Song},\ and\ \citenamefont
  {Chen}}]{Liang-Liu-Song-Chen-2025}%
  \BibitemOpen
  \bibfield  {author} {\bibinfo {author} {\bibfnamefont {Z.}~\bibnamefont
  {Liang}}, \bibinfo {author} {\bibfnamefont {K.}~\bibnamefont {Liu}}, \bibinfo
  {author} {\bibfnamefont {H.}~\bibnamefont {Song}},\ and\ \bibinfo {author}
  {\bibfnamefont {Y.-A.}\ \bibnamefont {Chen}},\ }\bibfield  {title} {\bibinfo
  {title} {Generalized toric codes on twisted tori for quantum error
  correction},\ }\Eprint {https://arxiv.org/abs/2503.03827} {2503.03827}
  (\bibinfo {year} {2025}),\ \bibinfo {note} {[unpublished]}\BibitemShut
  {NoStop}%
\bibitem [{\citenamefont {Berthusen}\ \emph {et~al.}(2025)\citenamefont
  {Berthusen}, \citenamefont {Devulapalli}, \citenamefont {Schoute},
  \citenamefont {Childs}, \citenamefont {Gullans}, \citenamefont {Gorshkov},\
  and\ \citenamefont {Gottesman}}]{Berthusen-etal-Gottesman-2025}%
  \BibitemOpen
  \bibfield  {author} {\bibinfo {author} {\bibfnamefont {N.}~\bibnamefont
  {Berthusen}}, \bibinfo {author} {\bibfnamefont {D.}~\bibnamefont
  {Devulapalli}}, \bibinfo {author} {\bibfnamefont {E.}~\bibnamefont
  {Schoute}}, \bibinfo {author} {\bibfnamefont {A.~M.}\ \bibnamefont {Childs}},
  \bibinfo {author} {\bibfnamefont {M.~J.}\ \bibnamefont {Gullans}}, \bibinfo
  {author} {\bibfnamefont {A.~V.}\ \bibnamefont {Gorshkov}},\ and\ \bibinfo
  {author} {\bibfnamefont {D.}~\bibnamefont {Gottesman}},\ }\bibfield  {title}
  {\bibinfo {title} {Toward a {2D} local implementation of quantum low-density
  parity-check codes},\ }\href {https://doi.org/10.1103/PRXQuantum.6.010306}
  {\bibfield  {journal} {\bibinfo  {journal} {PRX Quantum}\ }\textbf {\bibinfo
  {volume} {6}},\ \bibinfo {pages} {010306} (\bibinfo {year}
  {2025})}\BibitemShut {NoStop}%
\bibitem [{\citenamefont {Wang}\ \emph {et~al.}(2025)\citenamefont {Wang},
  \citenamefont {Lu}, \citenamefont {Zhang}, \citenamefont {Liu}, \citenamefont
  {Chen}, \citenamefont {Wang}, \citenamefont {Wu}, \citenamefont {Xu},
  \citenamefont {Zhu}, \citenamefont {Jin}, \citenamefont {Gao}, \citenamefont
  {Tan}, \citenamefont {Cui}, \citenamefont {Wang}, \citenamefont {Zou},
  \citenamefont {Zhang}, \citenamefont {Li}, \citenamefont {Shen},
  \citenamefont {Zhong}, \citenamefont {Bao}, \citenamefont {Zhu},
  \citenamefont {Han}, \citenamefont {He}, \citenamefont {Shen}, \citenamefont
  {Wang}, \citenamefont {Yang}, \citenamefont {Song}, \citenamefont {Deng},
  \citenamefont {Dong}, \citenamefont {Sun}, \citenamefont {Li}, \citenamefont
  {Ye}, \citenamefont {Jiang}, \citenamefont {Ma}, \citenamefont {Shen},
  \citenamefont {Zhang}, \citenamefont {Li}, \citenamefont {Guo}, \citenamefont
  {Wang}, \citenamefont {Song}, \citenamefont {Wang},\ and\ \citenamefont
  {Deng}}]{Wang-etal-Deng-2025}%
  \BibitemOpen
  \bibfield  {author} {\bibinfo {author} {\bibfnamefont {K.}~\bibnamefont
  {Wang}}, \bibinfo {author} {\bibfnamefont {Z.}~\bibnamefont {Lu}}, \bibinfo
  {author} {\bibfnamefont {C.}~\bibnamefont {Zhang}}, \bibinfo {author}
  {\bibfnamefont {G.}~\bibnamefont {Liu}}, \bibinfo {author} {\bibfnamefont
  {J.}~\bibnamefont {Chen}}, \bibinfo {author} {\bibfnamefont {Y.}~\bibnamefont
  {Wang}}, \bibinfo {author} {\bibfnamefont {Y.}~\bibnamefont {Wu}}, \bibinfo
  {author} {\bibfnamefont {S.}~\bibnamefont {Xu}}, \bibinfo {author}
  {\bibfnamefont {X.}~\bibnamefont {Zhu}}, \bibinfo {author} {\bibfnamefont
  {F.}~\bibnamefont {Jin}}, \bibinfo {author} {\bibfnamefont {Y.}~\bibnamefont
  {Gao}}, \bibinfo {author} {\bibfnamefont {Z.}~\bibnamefont {Tan}}, \bibinfo
  {author} {\bibfnamefont {Z.}~\bibnamefont {Cui}}, \bibinfo {author}
  {\bibfnamefont {N.}~\bibnamefont {Wang}}, \bibinfo {author} {\bibfnamefont
  {Y.}~\bibnamefont {Zou}}, \bibinfo {author} {\bibfnamefont {A.}~\bibnamefont
  {Zhang}}, \bibinfo {author} {\bibfnamefont {T.}~\bibnamefont {Li}}, \bibinfo
  {author} {\bibfnamefont {F.}~\bibnamefont {Shen}}, \bibinfo {author}
  {\bibfnamefont {J.}~\bibnamefont {Zhong}}, \bibinfo {author} {\bibfnamefont
  {Z.}~\bibnamefont {Bao}}, \bibinfo {author} {\bibfnamefont {Z.}~\bibnamefont
  {Zhu}}, \bibinfo {author} {\bibfnamefont {Y.}~\bibnamefont {Han}}, \bibinfo
  {author} {\bibfnamefont {Y.}~\bibnamefont {He}}, \bibinfo {author}
  {\bibfnamefont {J.}~\bibnamefont {Shen}}, \bibinfo {author} {\bibfnamefont
  {H.}~\bibnamefont {Wang}}, \bibinfo {author} {\bibfnamefont {J.-N.}\
  \bibnamefont {Yang}}, \bibinfo {author} {\bibfnamefont {Z.}~\bibnamefont
  {Song}}, \bibinfo {author} {\bibfnamefont {J.}~\bibnamefont {Deng}}, \bibinfo
  {author} {\bibfnamefont {H.}~\bibnamefont {Dong}}, \bibinfo {author}
  {\bibfnamefont {Z.-Z.}\ \bibnamefont {Sun}}, \bibinfo {author} {\bibfnamefont
  {W.}~\bibnamefont {Li}}, \bibinfo {author} {\bibfnamefont {Q.}~\bibnamefont
  {Ye}}, \bibinfo {author} {\bibfnamefont {S.}~\bibnamefont {Jiang}}, \bibinfo
  {author} {\bibfnamefont {Y.}~\bibnamefont {Ma}}, \bibinfo {author}
  {\bibfnamefont {P.-X.}\ \bibnamefont {Shen}}, \bibinfo {author}
  {\bibfnamefont {P.}~\bibnamefont {Zhang}}, \bibinfo {author} {\bibfnamefont
  {H.}~\bibnamefont {Li}}, \bibinfo {author} {\bibfnamefont {Q.}~\bibnamefont
  {Guo}}, \bibinfo {author} {\bibfnamefont {Z.}~\bibnamefont {Wang}}, \bibinfo
  {author} {\bibfnamefont {C.}~\bibnamefont {Song}}, \bibinfo {author}
  {\bibfnamefont {H.}~\bibnamefont {Wang}},\ and\ \bibinfo {author}
  {\bibfnamefont {D.-L.}\ \bibnamefont {Deng}},\ }\bibfield  {title} {\bibinfo
  {title} {Demonstration of low-overhead quantum error correction codes},\
  }\Eprint {https://arxiv.org/abs/2505.09684} {2505.09684}  (\bibinfo {year}
  {2025}),\ \bibinfo {note} {[unpublished]}\BibitemShut {NoStop}%
\bibitem [{\citenamefont {Aasen}\ \emph
  {et~al.}(2025{\natexlab{a}})\citenamefont {Aasen}, \citenamefont {Haah},
  \citenamefont {Hastings},\ and\ \citenamefont
  {Wang}}]{Aasen-Haah-Hastings-Wang-2025}%
  \BibitemOpen
  \bibfield  {author} {\bibinfo {author} {\bibfnamefont {D.}~\bibnamefont
  {Aasen}}, \bibinfo {author} {\bibfnamefont {J.}~\bibnamefont {Haah}},
  \bibinfo {author} {\bibfnamefont {M.~B.}\ \bibnamefont {Hastings}},\ and\
  \bibinfo {author} {\bibfnamefont {Z.}~\bibnamefont {Wang}},\ }\bibfield
  {title} {\bibinfo {title} {Geometrically enhanced topological quantum
  codes},\ }\Eprint {https://arxiv.org/abs/2505.10403} {2505.10403}  (\bibinfo
  {year} {2025}{\natexlab{a}}),\ \bibinfo {note} {unpublished}\BibitemShut
  {NoStop}%
\bibitem [{\citenamefont {Aasen}\ \emph
  {et~al.}(2025{\natexlab{b}})\citenamefont {Aasen}, \citenamefont {Hastings},
  \citenamefont {Kliuchnikov}, \citenamefont {Bello-Rivas}, \citenamefont
  {Paetznick}, \citenamefont {Chao}, \citenamefont {Reichardt}, \citenamefont
  {Zanner}, \citenamefont {da~Silva}, \citenamefont {Wang},\ and\ \citenamefont
  {Svore}}]{Aasen-etal-Svore-2025}%
  \BibitemOpen
  \bibfield  {author} {\bibinfo {author} {\bibfnamefont {D.}~\bibnamefont
  {Aasen}}, \bibinfo {author} {\bibfnamefont {M.~B.}\ \bibnamefont {Hastings}},
  \bibinfo {author} {\bibfnamefont {V.}~\bibnamefont {Kliuchnikov}}, \bibinfo
  {author} {\bibfnamefont {J.~M.}\ \bibnamefont {Bello-Rivas}}, \bibinfo
  {author} {\bibfnamefont {A.}~\bibnamefont {Paetznick}}, \bibinfo {author}
  {\bibfnamefont {R.}~\bibnamefont {Chao}}, \bibinfo {author} {\bibfnamefont
  {B.~W.}\ \bibnamefont {Reichardt}}, \bibinfo {author} {\bibfnamefont
  {M.}~\bibnamefont {Zanner}}, \bibinfo {author} {\bibfnamefont {M.~P.}\
  \bibnamefont {da~Silva}}, \bibinfo {author} {\bibfnamefont {Z.}~\bibnamefont
  {Wang}},\ and\ \bibinfo {author} {\bibfnamefont {K.~M.}\ \bibnamefont
  {Svore}},\ }\bibfield  {title} {\bibinfo {title} {A topologically
  fault-tolerant quantum computer with four dimensional geometric codes},\
  }\Eprint {https://arxiv.org/abs/2506.15130} {2506.15130}  (\bibinfo {year}
  {2025}{\natexlab{b}}),\ \bibinfo {note} {unpublished}\BibitemShut {NoStop}%
\bibitem [{\citenamefont {Tillich}\ and\ \citenamefont
  {Z{\'e}mor}(2009)}]{Tillich-Zemor-2009}%
  \BibitemOpen
  \bibfield  {author} {\bibinfo {author} {\bibfnamefont {J.-P.}\ \bibnamefont
  {Tillich}}\ and\ \bibinfo {author} {\bibfnamefont {G.}~\bibnamefont
  {Z{\'e}mor}},\ }\bibfield  {title} {\bibinfo {title} {Quantum {LDPC} codes
  with positive rate and minimum distance proportional to {$\sqrt{n}$}},\ }in\
  \href {https://doi.org/10.1109/ISIT.2009.5205648} {\emph {\bibinfo
  {booktitle} {Proc. IEEE Int. Symp. Inf. Theory (ISIT)}}}\ (\bibinfo {year}
  {2009})\ pp.\ \bibinfo {pages} {799--803}\BibitemShut {NoStop}%
\bibitem [{\citenamefont {Borello}\ \emph {et~al.}(2022)\citenamefont
  {Borello}, \citenamefont {De~La~Cruz},\ and\ \citenamefont
  {Willems}}]{Borello-delaCruz-Willems-2022}%
  \BibitemOpen
  \bibfield  {author} {\bibinfo {author} {\bibfnamefont {M.}~\bibnamefont
  {Borello}}, \bibinfo {author} {\bibfnamefont {J.}~\bibnamefont
  {De~La~Cruz}},\ and\ \bibinfo {author} {\bibfnamefont {W.}~\bibnamefont
  {Willems}},\ }\bibfield  {title} {\bibinfo {title} {On checkable codes in
  group algebras},\ }\href {https://doi.org/10.1142/S0219498822501250}
  {\bibfield  {journal} {\bibinfo  {journal} {Journal of Algebra and Its
  Applications}\ }\textbf {\bibinfo {volume} {21}},\ \bibinfo {pages} {2250125}
  (\bibinfo {year} {2022})},\ \Eprint {https://arxiv.org/abs/1901.10979}
  {1901.10979} \BibitemShut {NoStop}%
\bibitem [{\citenamefont {Calderbank}\ and\ \citenamefont
  {Shor}(1996)}]{Calderbank-Shor-1996}%
  \BibitemOpen
  \bibfield  {author} {\bibinfo {author} {\bibfnamefont {A.~R.}\ \bibnamefont
  {Calderbank}}\ and\ \bibinfo {author} {\bibfnamefont {P.~W.}\ \bibnamefont
  {Shor}},\ }\bibfield  {title} {\bibinfo {title} {Good quantum
  error-correcting codes exist},\ }\href
  {https://doi.org/10.1103/PhysRevA.54.1098} {\bibfield  {journal} {\bibinfo
  {journal} {Phys. Rev. A}\ }\textbf {\bibinfo {volume} {54}},\ \bibinfo
  {pages} {1098} (\bibinfo {year} {1996})}\BibitemShut {NoStop}%
\bibitem [{\citenamefont {Steane}(1996)}]{Steane-1996}%
  \BibitemOpen
  \bibfield  {author} {\bibinfo {author} {\bibfnamefont {A.~M.}\ \bibnamefont
  {Steane}},\ }\bibfield  {title} {\bibinfo {title} {Simple quantum
  error-correcting codes},\ }\href {http://dx.doi.org/10.1103/PhysRevA.54.4741}
  {\bibfield  {journal} {\bibinfo  {journal} {Phys. Rev. A}\ }\textbf {\bibinfo
  {volume} {54}},\ \bibinfo {pages} {4741} (\bibinfo {year}
  {1996})}\BibitemShut {NoStop}%
\bibitem [{Note1()}]{Note1}%
  \BibitemOpen
  \bibinfo {note} {Commuting matrices can be constructed for any pair of group
  algebra elements even for a non-abelian group\cite {Panteleev-Kalachev-2021}.
  While there is no direct generalization for arbitrary sets of $D>2$
  non-abelian group algebra elements, less general constructions do exist.
  These go outside the scope of the present work.}\BibitemShut {Stop}%
\bibitem [{Note2()}]{Note2}%
  \BibitemOpen
  \bibinfo {note} {Such a decomposition exists according to the Fundamental
  Theorem of finite abelian groups.}\BibitemShut {Stop}%
\bibitem [{Note3()}]{Note3}%
  \BibitemOpen
  \bibinfo {note} {Otherwise, the spaces in an AMC complex can be decomposed
  onto a direct sum of subspaces corresponding to cosets of the subgroup
  $G_X=\langle \{x_j\}\rangle $ in the original group $G$; see Sec.\ IV-C in
  Ref.~\protect \rev@citealp {Lin-Pryadko-2023}. In the abelian case these
  complexes are permutation-equivalent to a complex in the subgroup $G_X\le
  G$.}\BibitemShut {Stop}%
\bibitem [{Note4()}]{Note4}%
  \BibitemOpen
  \bibinfo {note} {Note, however, that the circuit distance need not always
  coincide with that of the original code. As an example, while any
  single-ancilla measurement circuit would work for a toric code\cite
  {Manes-Claes-2025}, rotated surface codes require a carefully designed
  {\protect \sf И}-{\protect \sf Z}\ addressing scheme for
  fault-tolerance\cite {Tomita-Svore-2014}.}\BibitemShut {Stop}%
\bibitem [{\citenamefont {Postema}\ and\ \citenamefont
  {Kokkelmans}(2025)}]{Postema-Kokkelmans-2025}%
  \BibitemOpen
  \bibfield  {author} {\bibinfo {author} {\bibfnamefont {J.~J.}\ \bibnamefont
  {Postema}}\ and\ \bibinfo {author} {\bibfnamefont {S.~J. J. M.~F.}\
  \bibnamefont {Kokkelmans}},\ }\bibfield  {title} {\bibinfo {title} {Existence
  and characterisation of bivariate bicycle codes},\ }\Eprint
  {https://arxiv.org/abs/2502.17052} {2502.17052}  (\bibinfo {year} {2025}),\
  \bibinfo {note} {unpublished}\BibitemShut {NoStop}%
\bibitem [{\citenamefont {Panteleev}\ and\ \citenamefont
  {Kalachev}(2021{\natexlab{a}})}]{Panteleev-Kalachev-2019}%
  \BibitemOpen
  \bibfield  {author} {\bibinfo {author} {\bibfnamefont {P.}~\bibnamefont
  {Panteleev}}\ and\ \bibinfo {author} {\bibfnamefont {G.}~\bibnamefont
  {Kalachev}},\ }\bibfield  {title} {\bibinfo {title} {Degenerate quantum
  {LDPC} codes with good finite length performance},\ }\href
  {https://doi.org/10.22331/q-2021-11-22-585} {\bibfield  {journal} {\bibinfo
  {journal} {Quantum}\ }\textbf {\bibinfo {volume} {5}},\ \bibinfo {pages}
  {585} (\bibinfo {year} {2021}{\natexlab{a}})},\ \Eprint
  {https://arxiv.org/abs/1904.02703} {1904.02703} \BibitemShut {NoStop}%
\bibitem [{\citenamefont {Drozd}\ and\ \citenamefont
  {Kirichenko}(1994)}]{Drozd-Kirichenko-book-1994}%
  \BibitemOpen
  \bibfield  {author} {\bibinfo {author} {\bibfnamefont {Y.~A.}\ \bibnamefont
  {Drozd}}\ and\ \bibinfo {author} {\bibfnamefont {V.~V.}\ \bibnamefont
  {Kirichenko}},\ }\href {https://doi.org/10.1007/978-3-642-76244-4} {\emph
  {\bibinfo {title} {Finite Dimensional Algebras}}}\ (\bibinfo  {publisher}
  {Springer-Verlag},\ \bibinfo {year} {1994})\BibitemShut {NoStop}%
\bibitem [{\citenamefont {Panteleev}\ and\ \citenamefont
  {Kalachev}(2022)}]{Panteleev-Kalachev-2020}%
  \BibitemOpen
  \bibfield  {author} {\bibinfo {author} {\bibfnamefont {P.}~\bibnamefont
  {Panteleev}}\ and\ \bibinfo {author} {\bibfnamefont {G.}~\bibnamefont
  {Kalachev}},\ }\bibfield  {title} {\bibinfo {title} {Quantum {LDPC} codes
  with almost linear minimum distance},\ }\href
  {https://doi.org/10.1109/TIT.2021.3119384} {\bibfield  {journal} {\bibinfo
  {journal} {IEEE Transactions on Information Theory}\ }\textbf {\bibinfo
  {volume} {68}},\ \bibinfo {pages} {213} (\bibinfo {year} {2022})},\ \Eprint
  {https://arxiv.org/abs/arXiv:2012.04068} {arXiv:2012.04068} \BibitemShut
  {NoStop}%
\bibitem [{\citenamefont {Fan}\ and\ \citenamefont {Lin}(2021)}]{Fan-Lin-2021}%
  \BibitemOpen
  \bibfield  {author} {\bibinfo {author} {\bibfnamefont {Y.}~\bibnamefont
  {Fan}}\ and\ \bibinfo {author} {\bibfnamefont {L.}~\bibnamefont {Lin}},\
  }\bibfield  {title} {\bibinfo {title} {Dihedral group codes over finite
  fields},\ }\href {https://doi.org/10.1109/TIT.2021.3088457} {\bibfield
  {journal} {\bibinfo  {journal} {IEEE Transactions on Information Theory}\
  }\textbf {\bibinfo {volume} {67}},\ \bibinfo {pages} {5016} (\bibinfo {year}
  {2021})}\BibitemShut {NoStop}%
\bibitem [{\citenamefont {Bazzi}\ and\ \citenamefont
  {Mitter}(2006)}]{Bazzi-Mitter-2006}%
  \BibitemOpen
  \bibfield  {author} {\bibinfo {author} {\bibfnamefont {L.~M.}\ \bibnamefont
  {Bazzi}}\ and\ \bibinfo {author} {\bibfnamefont {S.~K.}\ \bibnamefont
  {Mitter}},\ }\bibfield  {title} {\bibinfo {title} {Some randomized code
  constructions from group actions},\ }\href
  {https://doi.org/10.1109/TIT.2006.876244} {\bibfield  {journal} {\bibinfo
  {journal} {IEEE Transactions on Information Theory}\ }\textbf {\bibinfo
  {volume} {52}},\ \bibinfo {pages} {3210} (\bibinfo {year}
  {2006})}\BibitemShut {NoStop}%
\bibitem [{\citenamefont {Haviv}\ \emph {et~al.}(2017)\citenamefont {Haviv},
  \citenamefont {Langberg}, \citenamefont {Schwartz},\ and\ \citenamefont
  {Yaakobi}}]{Haviv-Langberg-Schwartz-Yaakobi-2017}%
  \BibitemOpen
  \bibfield  {author} {\bibinfo {author} {\bibfnamefont {I.}~\bibnamefont
  {Haviv}}, \bibinfo {author} {\bibfnamefont {M.}~\bibnamefont {Langberg}},
  \bibinfo {author} {\bibfnamefont {M.}~\bibnamefont {Schwartz}},\ and\
  \bibinfo {author} {\bibfnamefont {E.}~\bibnamefont {Yaakobi}},\ }\bibfield
  {title} {\bibinfo {title} {Non-linear cyclic codes that attain the
  {G}ilbert-{V}arshamov bound},\ }in\ \href
  {https://doi.org/10.1109/ISIT.2017.8006595} {\emph {\bibinfo {booktitle}
  {2017 {IEEE} International Symposium on Information Theory (ISIT)}}}\
  (\bibinfo {year} {2017})\ pp.\ \bibinfo {pages} {586--588}\BibitemShut
  {NoStop}%
\bibitem [{\citenamefont {Lin}\ and\ \citenamefont
  {Weldon}(1967)}]{Lin-Weldon-1967}%
  \BibitemOpen
  \bibfield  {author} {\bibinfo {author} {\bibfnamefont {S.}~\bibnamefont
  {Lin}}\ and\ \bibinfo {author} {\bibfnamefont {E.~J.}\ \bibnamefont
  {Weldon}},\ }\bibfield  {title} {\bibinfo {title} {Long {BCH} codes are
  bad},\ }\href {https://doi.org/https://doi.org/10.1016/S0019-9958(67)90660-2}
  {\bibfield  {journal} {\bibinfo  {journal} {Information and Control}\
  }\textbf {\bibinfo {volume} {11}},\ \bibinfo {pages} {445} (\bibinfo {year}
  {1967})}\BibitemShut {NoStop}%
\bibitem [{\citenamefont {Berlekamp}\ and\ \citenamefont
  {Justesen}(1974)}]{Berlekamp-Justesen-1974}%
  \BibitemOpen
  \bibfield  {author} {\bibinfo {author} {\bibfnamefont {E.}~\bibnamefont
  {Berlekamp}}\ and\ \bibinfo {author} {\bibfnamefont {J.}~\bibnamefont
  {Justesen}},\ }\bibfield  {title} {\bibinfo {title} {Some long cyclic linear
  binary codes are not so bad},\ }\href
  {https://doi.org/10.1109/TIT.1974.1055222} {\bibfield  {journal} {\bibinfo
  {journal} {IEEE Transactions on Information Theory}\ }\textbf {\bibinfo
  {volume} {20}},\ \bibinfo {pages} {351} (\bibinfo {year} {1974})}\BibitemShut
  {NoStop}%
\bibitem [{\citenamefont {Pryadko}\ and\ \citenamefont
  {Zeng}(2024)}]{Pryadko-2025-distm4ri}%
  \BibitemOpen
  \bibfield  {author} {\bibinfo {author} {\bibfnamefont {L.~P.}\ \bibnamefont
  {Pryadko}}\ and\ \bibinfo {author} {\bibfnamefont {W.}~\bibnamefont {Zeng}},\
  }\href@noop {} {\bibinfo {title} {dist-m4ri--- distance of a classical or
  quantum {CSS} code}},\ \bibinfo {howpublished}
  {\url{https://github.com/QEC-pages/dist-m4ri}} (\bibinfo {year}
  {2024})\BibitemShut {NoStop}%
\bibitem [{\citenamefont {Tomita}\ and\ \citenamefont
  {Svore}(2014)}]{Tomita-Svore-2014}%
  \BibitemOpen
  \bibfield  {author} {\bibinfo {author} {\bibfnamefont {Y.}~\bibnamefont
  {Tomita}}\ and\ \bibinfo {author} {\bibfnamefont {K.~M.}\ \bibnamefont
  {Svore}},\ }\bibfield  {title} {\bibinfo {title} {Low-distance surface codes
  under realistic quantum noise},\ }\href
  {https://doi.org/10.1103/PhysRevA.90.062320} {\bibfield  {journal} {\bibinfo
  {journal} {Phys. Rev. A}\ }\textbf {\bibinfo {volume} {90}},\ \bibinfo
  {pages} {062320} (\bibinfo {year} {2014})},\ \Eprint
  {https://arxiv.org/abs/1404.3747} {1404.3747} \BibitemShut {NoStop}%
\bibitem [{Note5()}]{Note5}%
  \BibitemOpen
  \bibinfo {note} {To guarantee rank preservation after removal of a block row
  (which is important for the {\protect \tt 1111} cycle), guided by Lemma~\ref
  {th:lemma-A-factor}, the polynomials were ordered so that $a_4=1+x$ [which
  gives block $D$ in Eqs.~(\ref {eq:mbc4-2}) and (\ref
  {eq:mbc4-3})].}\BibitemShut {Stop}%
\bibitem [{\citenamefont {Gidney}(2021)}]{Gidney-2021-stim}%
  \BibitemOpen
  \bibfield  {author} {\bibinfo {author} {\bibfnamefont {C.}~\bibnamefont
  {Gidney}},\ }\bibfield  {title} {\bibinfo {title} {Stim: a fast stabilizer
  circuit simulator},\ }\href {https://doi.org/10.22331/q-2021-07-06-497}
  {\bibfield  {journal} {\bibinfo  {journal} {{Quantum}}\ }\textbf {\bibinfo
  {volume} {5}},\ \bibinfo {pages} {497} (\bibinfo {year} {2021})}\BibitemShut
  {NoStop}%
\bibitem [{\citenamefont {Pryadko}(2025)}]{Pryadko-2025-vecdec}%
  \BibitemOpen
  \bibfield  {author} {\bibinfo {author} {\bibfnamefont {L.~P.}\ \bibnamefont
  {Pryadko}},\ }\href@noop {} {\bibinfo {title} {vecdec --- vectorized decoder
  and {LER} estimator}},\ \bibinfo {howpublished}
  {\url{https://github.com/QEC-pages/vecdec}} (\bibinfo {year}
  {2025})\BibitemShut {NoStop}%
\bibitem [{\citenamefont {Dumer}\ \emph {et~al.}(2017)\citenamefont {Dumer},
  \citenamefont {Kovalev},\ and\ \citenamefont
  {Pryadko}}]{Dumer-Kovalev-Pryadko-IEEE-2017}%
  \BibitemOpen
  \bibfield  {author} {\bibinfo {author} {\bibfnamefont {I.}~\bibnamefont
  {Dumer}}, \bibinfo {author} {\bibfnamefont {A.~A.}\ \bibnamefont {Kovalev}},\
  and\ \bibinfo {author} {\bibfnamefont {L.~P.}\ \bibnamefont {Pryadko}},\
  }\bibfield  {title} {\bibinfo {title} {Distance verification for classical
  and quantum {LDPC} codes},\ }\href {https://doi.org/10.1109/TIT.2017.2690381}
  {\bibfield  {journal} {\bibinfo  {journal} {IEEE Trans. Inf. Th.}\ }\textbf
  {\bibinfo {volume} {63}},\ \bibinfo {pages} {4675} (\bibinfo {year}
  {2017})}\BibitemShut {NoStop}%
\bibitem [{Note6()}]{Note6}%
  \BibitemOpen
  \bibinfo {note} {This is unlike with the {\protect \tt BP+OSD} package\cite
  {roffe_decoding_2020,Roffe_LDPC_Python_tools_2022} whose performance may
  actually be degraded by additional detector events, see in Appendix A of
  Ref.~\protect \rev@citealp {Beni-Higgott-Shutty-2025}. In comparison, for the
  decoder used here, removing minority detector events increases logical error
  rates, e.g., from $0.20\%$ to $0.36\%$ at $p=0.1\%$ with the $[[42,6,4]]$
  code and the ``1212'' circuits.}\BibitemShut {Stop}%
\bibitem [{\citenamefont {Grospellier}\ \emph {et~al.}(2020)\citenamefont
  {Grospellier}, \citenamefont {Grou{\`e}s}, \citenamefont {Krishna},\ and\
  \citenamefont {Leverrier}}]{Grospellier-Groues-Krishna-Leverrier-2020}%
  \BibitemOpen
  \bibfield  {author} {\bibinfo {author} {\bibfnamefont {A.}~\bibnamefont
  {Grospellier}}, \bibinfo {author} {\bibfnamefont {L.}~\bibnamefont
  {Grou{\`e}s}}, \bibinfo {author} {\bibfnamefont {A.}~\bibnamefont
  {Krishna}},\ and\ \bibinfo {author} {\bibfnamefont {A.}~\bibnamefont
  {Leverrier}},\ }\bibfield  {title} {\bibinfo {title} {Combining hard and soft
  decoders for hypergraph product codes},\ }\Eprint
  {https://arxiv.org/abs/arXiv:2004.11199} {arXiv:2004.11199}  (\bibinfo {year}
  {2020}),\ \bibinfo {note} {unpublished}\BibitemShut {NoStop}%
\bibitem [{\citenamefont {Kovalev}\ and\ \citenamefont
  {Pryadko}(2013{\natexlab{b}})}]{Kovalev-Pryadko-FT-2013}%
  \BibitemOpen
  \bibfield  {author} {\bibinfo {author} {\bibfnamefont {A.~A.}\ \bibnamefont
  {Kovalev}}\ and\ \bibinfo {author} {\bibfnamefont {L.~P.}\ \bibnamefont
  {Pryadko}},\ }\bibfield  {title} {\bibinfo {title} {Fault tolerance of
  quantum low-density parity check codes with sublinear distance scaling},\
  }\href {https://doi.org/10.1103/PhysRevA.87.020304} {\bibfield  {journal}
  {\bibinfo  {journal} {Phys. Rev. A}\ }\textbf {\bibinfo {volume} {87}},\
  \bibinfo {pages} {020304(R)} (\bibinfo {year}
  {2013}{\natexlab{b}})}\BibitemShut {NoStop}%
\bibitem [{\citenamefont {Dumer}\ \emph {et~al.}(2015)\citenamefont {Dumer},
  \citenamefont {Kovalev},\ and\ \citenamefont
  {Pryadko}}]{Dumer-Kovalev-Pryadko-bnd-2015}%
  \BibitemOpen
  \bibfield  {author} {\bibinfo {author} {\bibfnamefont {I.}~\bibnamefont
  {Dumer}}, \bibinfo {author} {\bibfnamefont {A.~A.}\ \bibnamefont {Kovalev}},\
  and\ \bibinfo {author} {\bibfnamefont {L.~P.}\ \bibnamefont {Pryadko}},\
  }\bibfield  {title} {\bibinfo {title} {Thresholds for correcting errors,
  erasures, and faulty syndrome measurements in degenerate quantum codes},\
  }\href {https://doi.org/10.1103/PhysRevLett.115.050502} {\bibfield  {journal}
  {\bibinfo  {journal} {Phys. Rev. Lett.}\ }\textbf {\bibinfo {volume} {115}},\
  \bibinfo {pages} {050502} (\bibinfo {year} {2015})},\ \Eprint
  {https://arxiv.org/abs/1412.6172} {1412.6172} \BibitemShut {NoStop}%
\bibitem [{\citenamefont {Stephens}(2014)}]{Stephens-2014}%
  \BibitemOpen
  \bibfield  {author} {\bibinfo {author} {\bibfnamefont {A.~M.}\ \bibnamefont
  {Stephens}},\ }\bibfield  {title} {\bibinfo {title} {Fault-tolerant
  thresholds for quantum error correction with the surface code},\ }\href
  {https://doi.org/10.1103/PhysRevA.89.022321} {\bibfield  {journal} {\bibinfo
  {journal} {Phys. Rev. A}\ }\textbf {\bibinfo {volume} {89}},\ \bibinfo
  {pages} {022321} (\bibinfo {year} {2014})},\ \Eprint
  {https://arxiv.org/abs/arXiv:1311.5003} {arXiv:1311.5003} \BibitemShut
  {NoStop}%
\bibitem [{\citenamefont {Das}\ \emph {et~al.}(2022)\citenamefont {Das},
  \citenamefont {Locharla},\ and\ \citenamefont
  {Jones}}]{Das-Locharia-Jones-2022}%
  \BibitemOpen
  \bibfield  {author} {\bibinfo {author} {\bibfnamefont {P.}~\bibnamefont
  {Das}}, \bibinfo {author} {\bibfnamefont {A.}~\bibnamefont {Locharla}},\ and\
  \bibinfo {author} {\bibfnamefont {C.}~\bibnamefont {Jones}},\ }\bibfield
  {title} {\bibinfo {title} {{LILLIPUT}: a lightweight low-latency lookup-table
  based decoder for near-term quantum error correction},\ }in\ \href@noop {}
  {\emph {\bibinfo {booktitle} {Proc. 27th {ACM} {ASPLOS}}}}\ (\bibinfo
  {organization} {ACM},\ \bibinfo {address} {New York, NY, USA},\ \bibinfo
  {year} {2022})\ p.\ \bibinfo {pages} {541},\ \Eprint
  {https://arxiv.org/abs/arXiv:2108.06569} {arXiv:2108.06569} \BibitemShut
  {NoStop}%
\bibitem [{\citenamefont {Bombin}\ \emph {et~al.}(2013)\citenamefont {Bombin},
  \citenamefont {Chhajlany}, \citenamefont {Horodecki},\ and\ \citenamefont
  {Martin-Delgado}}]{Bombin-Chhajlany-Horodecki-MartinDelgado-2013}%
  \BibitemOpen
  \bibfield  {author} {\bibinfo {author} {\bibfnamefont {H.}~\bibnamefont
  {Bombin}}, \bibinfo {author} {\bibfnamefont {R.~W.}\ \bibnamefont
  {Chhajlany}}, \bibinfo {author} {\bibfnamefont {M.}~\bibnamefont
  {Horodecki}},\ and\ \bibinfo {author} {\bibfnamefont {M.~A.}\ \bibnamefont
  {Martin-Delgado}},\ }\bibfield  {title} {\bibinfo {title} {Self-correcting
  quantum computers},\ }\href {http://stacks.iop.org/1367-2630/15/i=5/a=055023}
  {\bibfield  {journal} {\bibinfo  {journal} {New Journal of Physics}\ }\textbf
  {\bibinfo {volume} {15}},\ \bibinfo {pages} {055023} (\bibinfo {year}
  {2013})}\BibitemShut {NoStop}%
\bibitem [{\citenamefont {Hastings}(2014)}]{Hastings-hyp-2013}%
  \BibitemOpen
  \bibfield  {author} {\bibinfo {author} {\bibfnamefont {M.~B.}\ \bibnamefont
  {Hastings}},\ }\bibfield  {title} {\bibinfo {title} {Decoding in hyperbolic
  spaces: {LDPC} codes with linear rate and efficient error correction},\
  }\href {https://doi.org/10.26421/QIC14.13-14-9} {\bibfield  {journal}
  {\bibinfo  {journal} {Quantum Information and Computation}\ }\textbf
  {\bibinfo {volume} {14}},\ \bibinfo {pages} {1187} (\bibinfo {year}
  {2014})},\ \Eprint {https://arxiv.org/abs/1312.2546} {1312.2546} \BibitemShut
  {NoStop}%
\bibitem [{\citenamefont {Panteleev}\ and\ \citenamefont
  {Kalachev}(2021{\natexlab{b}})}]{Panteleev-Kalachev-2021}%
  \BibitemOpen
  \bibfield  {author} {\bibinfo {author} {\bibfnamefont {P.}~\bibnamefont
  {Panteleev}}\ and\ \bibinfo {author} {\bibfnamefont {G.}~\bibnamefont
  {Kalachev}},\ }\bibfield  {title} {\bibinfo {title} {Asymptotically good
  quantum and locally testable classical {LDPC} codes},\ }\Eprint
  {https://arxiv.org/abs/arXiv:2111.03654} {arXiv:2111.03654}  (\bibinfo {year}
  {2021}{\natexlab{b}}),\ \bibinfo {note} {[Unpublished]}\BibitemShut {NoStop}%
\bibitem [{\citenamefont {Manes}\ and\ \citenamefont
  {Claes}(2025)}]{Manes-Claes-2025}%
  \BibitemOpen
  \bibfield  {author} {\bibinfo {author} {\bibfnamefont {A.~G.}\ \bibnamefont
  {Manes}}\ and\ \bibinfo {author} {\bibfnamefont {J.}~\bibnamefont {Claes}},\
  }\bibfield  {title} {\bibinfo {title} {Distance-preserving stabilizer
  measurements in hypergraph product codes},\ }\href
  {https://doi.org/10.22331/q-2025-01-30-1618} {\bibfield  {journal} {\bibinfo
  {journal} {Quantum}\ }\textbf {\bibinfo {volume} {9}},\ \bibinfo {pages}
  {1618} (\bibinfo {year} {2025})}\BibitemShut {NoStop}%
\bibitem [{\citenamefont {Roffe}\ \emph {et~al.}(2020)\citenamefont {Roffe},
  \citenamefont {White}, \citenamefont {Burton},\ and\ \citenamefont
  {Campbell}}]{roffe_decoding_2020}%
  \BibitemOpen
  \bibfield  {author} {\bibinfo {author} {\bibfnamefont {J.}~\bibnamefont
  {Roffe}}, \bibinfo {author} {\bibfnamefont {D.~R.}\ \bibnamefont {White}},
  \bibinfo {author} {\bibfnamefont {S.}~\bibnamefont {Burton}},\ and\ \bibinfo
  {author} {\bibfnamefont {E.}~\bibnamefont {Campbell}},\ }\bibfield  {title}
  {\bibinfo {title} {Decoding across the quantum low-density parity-check code
  landscape},\ }\href {https://doi.org/10.1103/physrevresearch.2.043423}
  {\bibfield  {journal} {\bibinfo  {journal} {Physical Review Research}\
  }\textbf {\bibinfo {volume} {2}},\ \bibinfo {pages} {043423} (\bibinfo {year}
  {2020})}\BibitemShut {NoStop}%
\bibitem [{\citenamefont {Roffe}(2022)}]{Roffe_LDPC_Python_tools_2022}%
  \BibitemOpen
  \bibfield  {author} {\bibinfo {author} {\bibfnamefont {J.}~\bibnamefont
  {Roffe}},\ }\href {https://pypi.org/project/ldpc/} {\bibinfo {title} {{LDPC:
  Python tools for low density parity check codes}}},\ \bibinfo {howpublished}
  {{PyPi} repository} (\bibinfo {year} {2022})\BibitemShut {NoStop}%
\bibitem [{\citenamefont {Beni}\ \emph {et~al.}(2025)\citenamefont {Beni},
  \citenamefont {Higgott},\ and\ \citenamefont
  {Shutty}}]{Beni-Higgott-Shutty-2025}%
  \BibitemOpen
  \bibfield  {author} {\bibinfo {author} {\bibfnamefont {L.~A.}\ \bibnamefont
  {Beni}}, \bibinfo {author} {\bibfnamefont {O.}~\bibnamefont {Higgott}},\ and\
  \bibinfo {author} {\bibfnamefont {N.}~\bibnamefont {Shutty}},\ }\bibfield
  {title} {\bibinfo {title} {Tesseract: A search-based decoder for quantum
  error correction},\ }\Eprint {https://arxiv.org/abs/2503.10988} {2503.10988}
  (\bibinfo {year} {2025}),\ \bibinfo {note} {unpublished}\BibitemShut
  {NoStop}%
\end{thebibliography}%

\end{document}